\documentclass[10pt, conference, letterpaper]{IEEEtran}
\usepackage{color,overpic,cite}
\usepackage{pdfsync}
\usepackage{url}
\usepackage{amsmath,amsfonts,amssymb,amsthm,bbm,bm,comment}
\newcommand{\ABT}{\textsf{ABT}}

\newcommand{\spyros}[1]{{#1}}
\newcommand{\eat}[1]{{}}
\newcommand{\mean}[1]{\mathbb{E}\!\left[#1\right]}
\newcommand{\prob}[1]{\mathbb{P}\!\left(#1\right)}

\newtheorem{theorem}{Theorem}
\newtheorem{lemma}[theorem]{Lemma}

\title{Placing Dynamic Content in Caches \\with Small Population}
\def\R{\mathbb{R}}
\def\N{\mathbb{N}}
\def\ind{\mathbbm{1}}

\def\Mcal{\mathcal{M}}

\def\Lcal{\mathcal{L}}

\def\PP{\mathbb{P}}
\def\EE{\mathbb{E}}

\author{
\IEEEauthorblockN{Mathieu Leconte, Georgios Paschos, Lazaros Gkatzikis, Moez Draief, Spyridon Vassilaras, Symeon Chouvardas}
\IEEEauthorblockA{Mathematical and Algorithmic Sciences Lab, France Research Center, Huawei Technologies Co., Ltd.\\
Email: firstname.lastname@huawei.com}
}

\begin{document}
\maketitle
\begin{abstract}
This paper addresses a fundamental limitation for the adoption of caching for wireless access networks due to small population sizes. This shortcoming is due to two main challenges: (i) making timely estimates of varying content popularity and (ii) inferring popular content from small samples. We propose a framework which alleviates such limitations. 

To timely estimate varying popularity in a context of a single cache we propose an Age-Based Threshold ($\ABT$) policy which caches all contents requested more times than a threshold $\widetilde N(\tau)$, where $\tau$ is the content age.
We show that $\ABT$ is asymptotically hit rate optimal in the many contents regime, which allows us to obtain the first characterization of the optimal performance of a caching system in a dynamic context. 
We then address small sample sizes focusing on $L$ local caches and one global cache. On the one hand we show that the global cache learns $L$ times faster by aggregating all requests from local caches, which improves hit rates. On the other hand, aggregation washes out local characteristics of correlated traffic which penalizes hit rate.
This motivates coordination mechanisms which combine global learning of popularity scores in clusters and LRU with prefetching.  
\end{abstract}
\vspace{-0.2in}
\section{Introduction}
Content Delivery Networks (CDNs) transformed the way data is replicated to respond to the ever increasing demand for popular content. The underlying technology uses large network caches to cover densely populated areas of millions of users. This paradigm yields to a number of benefits for the performance of the network, namely reducing latency and saving bandwidth. To further gain on such metrics in a wireless network, content can be stored closer yet to the user, e.g.,~at a base station or at the mobile. The main limitation to the adoption of such an appealing approach is that small local caches will only have a partial view of the content dynamics and each cache will only see populations of small sizes around it. There is indeed a general consensus that small population sizes result in poor hit rates. 
This paper contributes to answering the following fundamental question: \emph{How can one achieve good hit rates in caches which cover small populations?}

Fresh content such as news, music or TV series is produced on a regular basis. One of the characteristics of such content is that it is ephemeral, i.e., it is highly demanded for a certain duration and then the demand fades. Tracking an ever changing popularity profile of content is challenging as effective caching crucially relies on the knowledge of content popularity. 
This aspect is deemed one of the main hurdles to deploying caches closer to the user. At first sight, tracking popularity for local caches may seem hopeless due to the small sample size. \eat{Although content popularity evolution in time (a phenomenon known as temporal locality) and popularity distribution correlations among different locations (known as geographical locality) have been well documented in the literature, a large percentage of the caching literature uses the (time and location) Independent Request Model (IRM) as a basis for algorithmic design and performance analysis.}   
In this paper, we study caches with small population under the assumption of time-varying and unknown content popularity. In what follows we describe an important application of our research.

\begin{figure}[t]
\begin{center}
   \hspace{-0.6in}\begin{overpic}[scale=0.22]{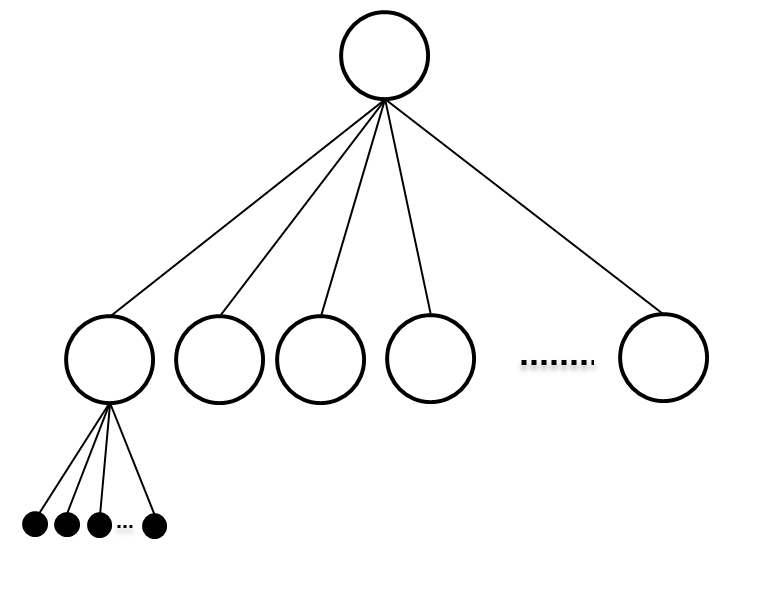}
        \put(26,7){\small small user population}
            \put(97,30){\small $L$ local caches}
            \put(60,70){\small global cache}
      \end{overpic}\vspace{-0.2in} 
      \caption{A hierarchy of $L$ local caches, each one receiving requests from a small population, and a global cache which receives requests from the aggregate large population.}
            \label{fig:distr}
						\vspace{-0.4in}
            \end{center}
\end{figure}

\subsection{Reducing Content Latency with Base Station Caches}
\label{sec:motivation}

The upcoming 5G wireless architectures pose stringent requirements in terms of latency \cite{5G} and motivate the placement of content near the user \cite{5Gcaching}.
Introducing caches at the network edge is an appealing solution since the cost of network equipment (base station or user equipment) substantially exceeds the cost of installing a cache \cite{Roberts13}.
There has recently been a large body of work on cache optimization for wireless systems, cf.~\cite{FemtoCaching, J_gitzenis_13, JiMolisch13, Niesen13}. However, all these ideas suffer from two main unrealistic assumptions, (i) the cache size is of the order of the catalog size, and (ii) the content popularities are known (or static). In this paper we study caching at wireless access by relaxing these two assumptions. Below we discuss the context of our paper.

Regarding the cache size, we remark that 
a determining factor for the caching performance is the ratio $C/M$, where $M$ is the size of the content catalog and $C$ is the cache size.\footnote{In this paper we will make the common simplifying assumption that all files have the same size\spyros {, which is well justified by the fact that we can break large files into equal size chunks and consider the chunks as the cacheable contents}. Hence, $C$ denotes the number of the contents which can be cached.} 
Prior studies of caching performance have shown that when $C/M$ is small, the probability to find a content in the cache becomes negligible~\cite{Roberts13}.
Since the base station (or mobile) cache is physically constrained, caches must be relatively small in storage size, and hence ineffective.

\spyros {There is one important case, however, where the effect of small cache size is counterbalanced by decreasing the size of cached contents: When content access latency reduction is the primary objective of caching, \emph{contents can be split into small chunks and only a small fraction of these chunks need to be cached}} \cite{Sen99}. 
See Figure~\ref{fig:request} for an illustration of this technique.
An interesting observation is that the chunk hit probability (i.e., the probability of finding the first chunk in the cache) for cache size $C$  will be equal to the hit probability for cache size $\xi C$, where $\xi>1$ is the inverse of the file fraction which is required for smooth content display. Hence when using caching for latency minimization, the cache sizes are virtually scaled by $\xi$, which gives a solution to the challenge of small caches. In the remainder of the paper we will study  hit probabilities with large caches, with the understanding that this directly corresponds to  \emph{chunk} hit probabilities in small caches.

\begin{figure}[t]
      \centering
            \includegraphics[width=0.86\linewidth]{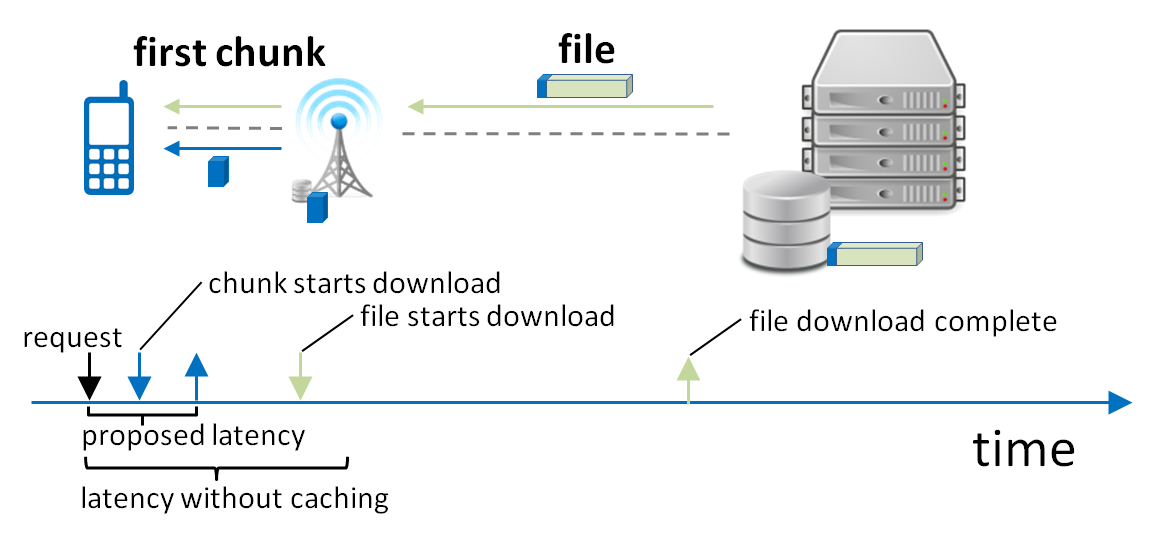}
						\vspace{-0.15in} 
      \caption{Latency reduction by means of caching the first content chunk at the base station.}\vspace{-0.4in}
			\label{fig:request}
\end{figure}

Having dealt with the small cache size issue, we shift our attention to the small population issue.
A base station cache receives requests from a small population of users, hence the number of requests per unit time is also very small. This in turn makes popularity estimation very challenging \cite{5Gcaching}. 
To make caching efficient for dynamic content popularity, 
 the remainder of the paper 
focuses on the study of  
caches with small population.
Below we survey related work on the topic and then explain our contribution.


\subsection{Related Work}

There is an increasing demand to use caching to combat the growth of mobile information \cite{IdealvsReality}, though the adaptation of caching techniques in the wireless domain is challenging \cite{5Gcaching}.
We deal with the problem of \emph{small cache size} using partial caching \cite{Ahle14} to improve the hit rate performance of the content header.
Storing contents partially has been previously proposed in the literature of proxy caching for streaming media applications \cite{Liu04,Sen99}.

The problem of \emph{small user population} on the other hand is relatively underexplored, perhaps because web caching is applied mostly in densely populated areas. It is noted though that hit rate performance has a sharp cutoff point as the user population becomes small \cite{Wolman}, mainly because there is insufficient room for correctly estimating content popularity. In this work we analyze the issue of small user population by formulating a problem on the interface of caching and learning.

Although the use of learning in the caching domain goes back to the days of prefetching for web caching, cf.~\cite{Pallis}, very recently it has been revived in the context of wireless networks. The use of transfer learning to tap social network side-information in order to alleviate data sparsity has been proposed in \cite{bastug2014b}, while \cite{Blasco} models the popularity learning as a multi-armed bandit problem. 
Most prior techniques are limited to the environment of the static popularity.
In practice, not only content popularity is dynamic, but moreover recent works argue that correct modeling has a significant impact on the performance analysis of caching schemes \cite{traverso2015,shen}. 
Learning time-varying popularity is actually an interesting problem on its own \cite{Szabo10,Ahmed13}.
However, any scheme which learns popularity separately from performing content placement is suboptimal \cite{moharir}, which  motivates a joint approach.

\subsection{Our Contribution}


For a single cache we study the joint problem of caching and learning time-varying popularity. We propose  a simple threshold policy called \spyros {Age-Based Threshold} ($\ABT$): a content is stored if it has been requested more than $\widetilde N(\tau)$ times, where $\tau$ is the age of the content\spyros {, i.e., the elapsed time since this content was first inserted in the contents catalog}, and $\widetilde N(\tau)$ is a selected threshold. 
We show that $\ABT$ is asymptotically optimal when we increase the number of contents, which provides a first characterization of the joint problem of caching and learning under time-varying popularity.


We then study an architecture where the popularity is estimated (or learned) at a global point which has access to all the requests arriving in $L$ local caches. 
 We prove that \emph{global learns faster}; by aggregating requests from all $L$ caches it is able to \spyros {track popularity changes} $L$ times faster.
If contents exhibit correlations in locations however, we show that the distribution of the local popularity of contents is more skewed, which means that local learning yields better performance  provided that these local popularities can be well estimated, i.e., \emph{local is more accurate}. \spyros {Combining the two last observations, we propose learning content popularities in clusters which are both able to retain local characteristics and to accumulate enough many request samples.}

Our goal is to learn a good estimate at the global point and then feed it back to the local caches in the form of content scores. In fact, we propose the modification of the threshold $\widetilde N(\tau)$ as a score which  
takes into account both the frequency of requests, as well as the content age. 
Using these scores we propose two globally coordinated mechanisms for \spyros {managing} the local caches: (i) a score-gated LRU\footnote{The Least Recently Used (LRU) cache replacement rule dictates that every requested content is cached, and if the cache is full then the least recently requested content is evicted. The score-gated counterpart avoids caching (gates) certain contents based on popularity scores. } and (ii) a score-based prefetching scheme.  
Here the term prefetching refers to \emph{the act of populating a cache with content which is not currently being requested at that cache}. 
We exhibit, using simulations, that prefetching is crucial for small population caches. 

Although global learning resolves the popularity estimation issue, there is a hidden outstanding issue in our architecture: the extra traffic incurred to prefetch content in the caches.
The latency minimization with caching comes at a cost of increased traffic in the core network due to prefetching. Using our proposed methodology, we evaluate this tradeoff by means of simulations and show that the incurred traffic can be kept significantly small--our simulations show 3\% of increase in total bandwidth in the worst case.



\section{Requests With Time-Varying Popularity}\label{sec:snm}

For our analysis we use a dynamic request model with time-varying popularities, the recently proposed Poisson \emph{Shot Noise Model} (SNM) \cite{traverso2015}.
This model introduces dynamicity in a simple manner while retaining the power law characteristics of instantaneous popularity observed from past works \cite{breslau99,adamic02,newman05}. In fact \cite{traverso2015} shows that SNM fits well real  data \spyros {of content requests in cellular networks}.

The lifetime of each content is associated with a shot, which is characterized by (i) a shape, (ii) a duration, (iii) an arrival instance, and (iv) a volume.
It is reported that the choices of (i)-(ii) have a smaller impact to the hit probability under the LRU cache management policy \cite{traverso2015}. Thus, in the following, we will consider rectangular pulses of fixed durations $T$ for all contents, see Figure~\ref{fig:SNM}, in order to develop an optimal cache management policy amenable to analytic expressions for the parameters of the policy and the resulting hit probability. For different shaped shots the details of our analysis must be revisited, but the main insights can be used to derive heuristic policies for the generic SNM. Also, one could perform the same analysis using a joint distribution of lifespan and shot volumes, as in \cite{olmos2014catalog}. 

The shot arrival times (iii) are points of a Poisson  process with constant rate $\lambda$. 
Denote with $\overline t_m$ the arrival time of shot $m$.
At time $t$ the alive content catalog is given by the set
\[
\mathcal{M}(t)=\{m: \overline t_m\leq t\leq \overline t_m+T\}.
\]

The shot volumes (iv) are determined by a power-law distribution, commonly known to fit well the instantaneous content popularity.
	More specifically, we set the request rate of content $m$ while it is alive to the random variable $\mu_m$ constructed in the following way. First, for any $m$, let $Z_m$ be an i.i.d. random variable drawn uniformly at random in $[0,1]$. Then
\[
\mu_m=\frac{Z_m^{-\alpha}}{\int_0^1z^{-\alpha}dz}\overline{\mu}=Z_m^{-\alpha}\overline{\mu}(1-\alpha), \text{ for all }m,
\]
where $\overline \mu$ is the mean popularity, and $\alpha$ is the power law exponent. 
{We let $f(x)$ denote the density of the power-law distribution at $x\in[\overline \mu(1-\alpha),+\infty)$; $f(x)$ is in fact a Pareto distribution \spyros {with parameters $\alpha_{Pareto} = 1/\alpha$ and $x_m = \overline{\mu}(1-\alpha)$} \cite{newman05}, which is the limit of Zipf distributions for large catalogs.\footnote{We choose to generate the power-law popularity in this particular way to facilitate the modeling in Section~\ref{sec:correlated_popularities}.}}
 Finally, we generate requests independently for each content $m$ using an independent Poisson process with rate $\mu_m$.


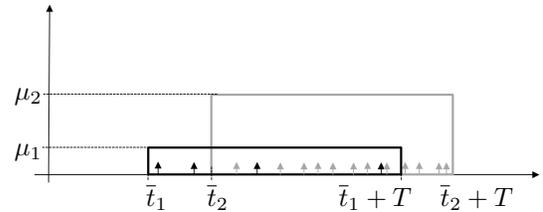
\begin{figure}[t!]
\begin{center}
   \begin{overpic}[scale=0.22]{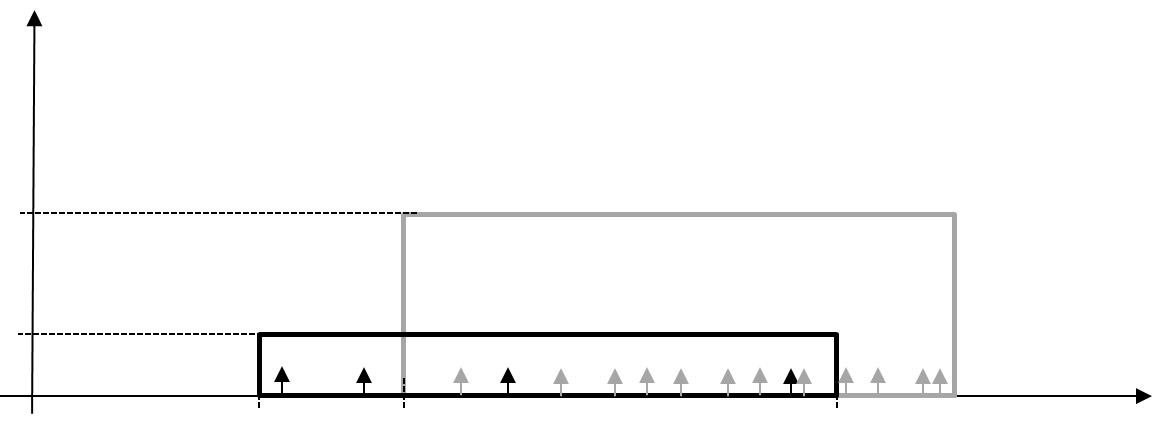}
    \put(22,-4){$\overline t_1$}
		\put(34,-4){$\overline t_2$}
		\put(60,-4){$\overline t_1+T$}
		\put(80,-4){$\overline t_2+T$}
		\put(-4,17){$\mu_2$}
		\put(-4,6){$\mu_1$}
   \end{overpic}
	\caption{Poisson Shot Noise Model (SNM). A realization showing two shots of different volumes and arrival time instances. In our model we keep the shot duration $T$ and the shot shape constant for all contents.}\label{fig:SNM}
	\vspace{-0.35in}
	\end{center}
\end{figure}

\section{Cache Hit Rate With Estimated Popularities}\label{sec:onecache}

In this section we focus on one cache receiving requests with dynamic and unknown popularities and we study the optimal hit rate performance.

\subsection{Hit Rate Optimization With Estimated Popularities}

In what follows we restrict attention to a caching controller which is not aware of the actual content popularities $\mu_m, m\in \mathcal{M}(t)$.
 Instead we assume that the controller 
estimates the content popularities via observations of past requests. To proceed with the analysis we additionally make the following simplifying assumptions:
\begin{itemize}
\item The controller knows the exact arrival times of shots $\overline t_m, m=1,2,\dots$.
\item There is no cost for replacing a content in the cache (i.e., there is no traffic cost for transmitting a content from the origin server to the cache or between caches). We will relax this assumption in section \ref{sec:prefetch}.
\end{itemize}

Our overarching goal is to maximize the hit rate  over the time horizon. However, given that the cache updates induce zero cost, we can decouple the time horizon hit rate optimization to individual problems of maximizing the  hit probability at each time instance. 

We may characterize the alive content $m\in\mathcal{M}(t)$ by its \emph{shot age} and the number of observed requests.
The shot age of content $m$ at time $t$ is the elapsed time since the content appeared in the system, denoted by $\tau_m(t)=t-\overline t_m$. 
	We denote with $N_m(t)$ the number of requests observed for content $m$ by time $t$.
	We represent the caching decision at time $t$ by a binary \emph{caching vector} $(y_m(t))_m$ of dimension $|\mathcal{M}(t)|$, where $y_m(t)=1$ if content $m$ is stored in the cache, and $y_m(t)=0$ otherwise. 
		Hereinafter we will exchangeably use the notations $(x_i)$ and $(x_i)_i$ to denote the vector $(x_i)_{i=1,\dots,I}=(x_1,\dots,x_I)$, and 
	drop the index whenever it is directly inferred from the context.
		The cache size $C$ dictates that the constraint $\sum_m y_m(t)\leq C$ must be satisfied at each time instance.
	
		By pointwise ergodicity of the SNM model  we may study any one time instance; we choose to study $t=0$.
	Since the caching performance at time $0$ depends only on alive shots $\Mcal (0)$, we only need to focus on random events in the time interval $[-T,0]$.
	To simplify notations, hereinafter we will omit the mention of the time index $0$ and write $\tau_m=\tau_m(0)$, $N_m=N_m(0)$, $\mathcal M=\mathcal M(0)$, and $y_m=y_m(0)$.

Next we would like to choose the caching vector $(y_m)$ to maximize the instantaneous hit probability at the origin. 
If the popularities were known the controller would employ the policy \emph{store the most popular}. However, the challenge here lies on the fact that the popularities $(\mu_m)$ are unknown. 
In fact, the instantaneous hit rate at the origin is given by $H(y)=\sum_{m\in\mathcal{M}}y_m\mu_m$. However, the popularities $\mu_m$ are not observed. Therefore, we need to consider instead the expected hit rate conditionally on the available information $(N_m),(\tau_m)$:
\[
\mean{H(y)\big|(N_m),(\tau_m)}=\sum_{m\in\mathcal{M}}y_m\mean{\mu_m|N_m,\tau_m},
\]
where the popularity estimates are computed using the prior model of $\mu_m$:
\begin{align}\label{eq:condexp}
\mean{\mu_m|N_m,\tau_m}
&=\frac{\int_{\mu_m} \mu_m\mathbbm P(N_m|\mu_m,\tau_m)f(\mu_m)d\mu_m}{\int_{\mu_m} \mathbbm{P}\left(N_m|\mu_m,\tau_m\right)f(\mu_m)d{\mu_m}},
\end{align}
where $f$ is the power-law density. To evaluate numerically \eqref{eq:condexp} observe that the popularity of content $m$ is equal to $\mu_m$ and constant  over the period $[-\tau_m,0]$, and the request process is Poisson, hence the term $\mathbbm P(N_m|\mu_m,\tau_m)$ is equal to the probability that a Poisson random variable with mean $\mu_m\tau_m$ is equal to $N_m$, i.e.,
\begin{align*}
\mathbbm P(N_m|\mu_m,\tau_m)&=\mathbbm P\left(Pois(\mu_m\tau_m)=N_m\right)\\
&=(\mu_m\tau_m)^{N_m}\frac{e^{-\mu_m\tau_m}}{N_m!}.
\end{align*}
For every instance $(N_m),(\tau_m)$, we want to find the best contents to store to maximize the conditional expected hit rate:

\vspace{0.1in}

\noindent \underline{\emph{Max instantaneous hit probability with estimated popularities:}}
\begin{equation}
y^*\big((N_m),(\tau_m)\big)=\arg\hspace{-0.15in}\max_{\hspace{-0.1in}\substack{\forall m,\:y_m\in \{0,1\}\\\sum\limits_{m\in\mathcal{M}}y_m=C}}\hspace{-0.05in}\sum_{\hspace{0.08in}m\in\mathcal{M}}y_m\mean{\mu_m|N_m,\tau_m}. \label{eq:opt1}
\end{equation}
The optimization \eqref{eq:opt1} can be solved by storing the $C$ items with the highest estimate $\mean{\mu_m|N_m,\tau_m}$. 
Given values for $(N_m),(\tau_m)$ we may compute numerically the instantaneous hit rate. However, the above are random. We define the maximum expected hit probability $h^*(\lambda,T)$ for shot arrival rate $\lambda$ and shot duration $T$, which will be our main performance metric:
\begin{equation}\label{eq:avhr}
h^*(\lambda,T)=\frac{\mean{H\left(y^*\big((N_m),(\tau_m)\big)\right)}}{\overline\mu\lambda T},
\end{equation}
where $|\Mcal|,(N_m),(\tau_m)$ all depend on $\lambda$ and $T$.
 Computing \eqref{eq:avhr} is complicated mainly because the set of active contents $\Mcal$ is itself random and the caching decisions are correlated across all the active contents. Below, we characterize a simple caching policy which is asymptotically optimal for large catalogs, which allows us to obtain an asymptotic expression for \eqref{eq:avhr}.


\subsection{Age-Based Threshold Policy}

\begin{figure}[t!]
\begin{center}
   \begin{overpic}[scale=0.44]{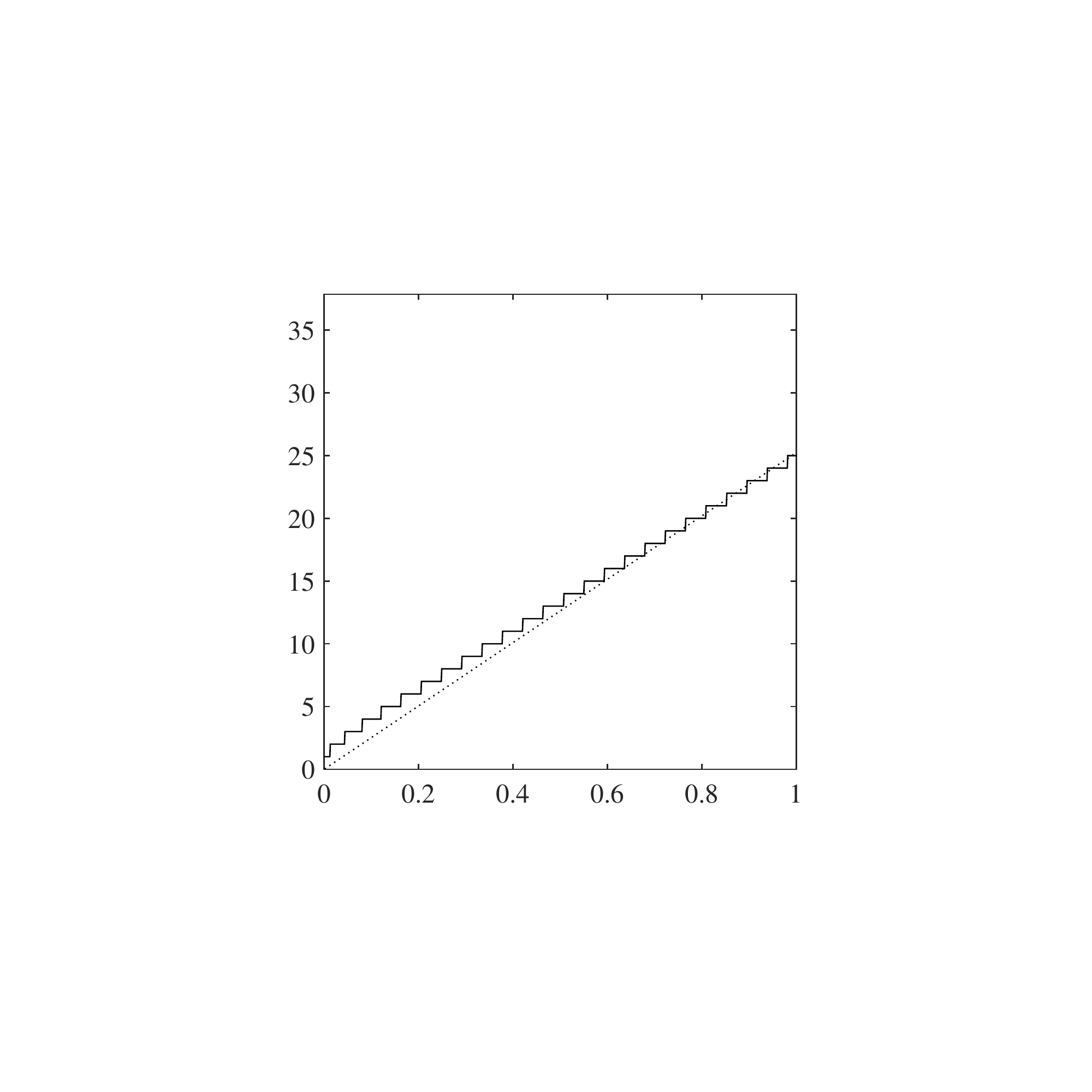}
        \put(35.5,-5.5){\footnotesize $\tau$ (shot age)}
				\put(32,88.5){\small $\ABT$ Threshold}
				\put(27,46){\scriptsize $\widetilde N(\tau)$}
				\put(35,14){\scriptsize $\text{const.}\times\tau$}
				\put(21,21.5){\vector(1,3){7}}
				\put(22,17.5){\vector(4,-1){12}}
   \end{overpic}
	 \begin{overpic}[scale=0.44]{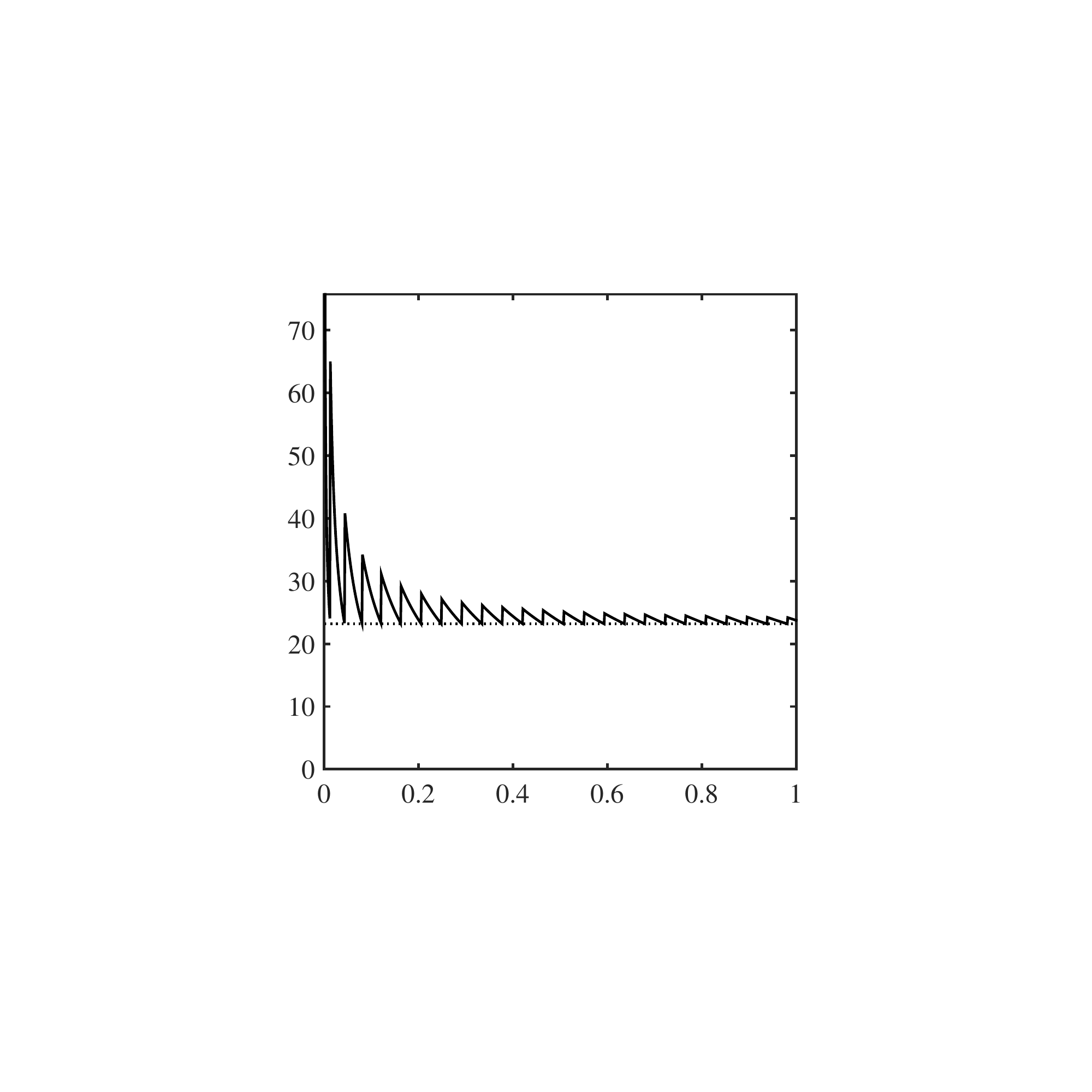}
        \put(35.5,-5.5){\footnotesize $\tau$ (shot age)}
				\put(11,89){\scriptsize Est.~popularity of marginally cached}
				\put(12.3,55.4){\vector(1,1){12}}
				\put(25,70){\tiny $\mean{\mu_m|N_m=\widetilde N(\tau),\tau_m=\tau}$}
				\put(11,35){\vector(1,-1){12}}
				\put(25,20){\scriptsize $\theta(\gamma_c)$}
   \end{overpic}
   \caption{$\ABT$ behavior in the many contents regime. \textbf{Parameters:} $\alpha=0.8$ (reported to be a typical value \cite{breslau99,adamic02}), $T=1$, $\overline\mu=20$, $\gamma_c=10\%$.\vspace{-0.6in}}
   \label{fig: optimal caching plot}
\end{center}
\end{figure}
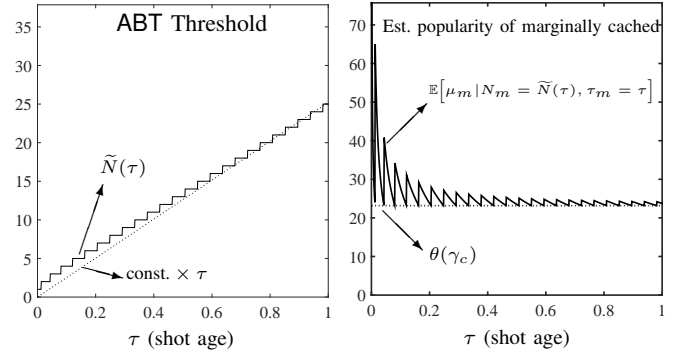

Our plan is to design a simplified caching policy $(y_m)$ which caches highly requested content without having to calculate the estimates $\widehat \mu_m=\mean{\mu_m|N_m,\tau_m}$ and to solve the optimization \eqref{eq:opt1} at every time instance. 
A complication comes from the fact that the shot age $\tau_m$ affects the estimate $\widehat \mu_m$. Intuitively, on average, to maximize hit probability we need to store more content with larger age $\tau_m$, because the uncertainty in the estimate $\widehat \mu_m$ is lower for them, which results in turn in fewer caching mistakes and a higher efficiency for older contents. We introduce a deterministic threshold $\widetilde N(\tau)$ which depends on the age $\tau$ and is used to allocate cache capacity differently for each $\tau$. Under our policy, content $m$ is stored if it satisfies \spyros { $N_m\geq\widetilde N(\tau_m)$}.


\noindent\rule{3.55in}{0.02in}

\noindent\textbf{Age-Based Threshold (\ABT) Policy.}

\vspace{-0.05in}
\noindent\rule{3.55in}{0.02in}

\noindent\textbf{Parameter Selection.} 
$C$ is the cache size (in contents), $\lambda$ is the shot rate, $T$ the shot duration, and hence $\lambda T$ is the average number of alive shots at any time instance.
Define $\gamma_c=C/\lambda T$, which is roughly the fraction of the content catalog which can be cached.
Denote with $F_{\widehat \mu_m}$ the \spyros {cumulative distribution function} of  $\widehat \mu_m$. 
Choose $\theta$ to be the $\gamma_c$-th upper-percentile of $F_{\widehat \mu_m}$, since $\widehat \mu_m$ has a density $F_{\widehat \mu_m}$ is invertible, hence
\[
\theta(\gamma_c)=F_{\widehat \mu_m}^{-1}(1-\gamma_c).
\]

\noindent\textbf{Age-Based Threshold.}
Choose the threshold $\tilde N(\tau)$ 
\begin{equation}\label{eq:threshold}
\widetilde N(\tau)=\min\{k\in \N : \mean{\mu_m|N_m=k,\tau_m=\tau}\geq \theta(\gamma_c)\}
\end{equation}

\noindent\textbf{Caching Vector.}
  For each content $m\in \mathcal{M}$ observe $N_m,\tau_m$ and choose: 
	\vspace{-0.1in}
\begin{align*}
y_m=\left\{\begin{array}{ll}
1 & \text{if } N_m\geq \widetilde N(\tau_m),\\
0 & \text{otherwise.}
\end{array}
\right.
\end{align*}
\noindent\textbf{Ensuring Cache Size Constraint.} If $\sum_m y_m>C$, then choose arbitrarily $\sum_m y_m-C$ contents and set $y_m=0$.

\noindent\rule{3.55in}{0.02in}

For a given content $m$, $\tau_m$ is known, hence under our policy the caching decision depends only on $N_m$, which simplifies greatly caching decisions. The  complicated part of the policy is to compute the threshold function $\widetilde N(\tau)$. However, this can be done in an offline manner: for any given parameters $C,\lambda,T$ and power law parameters $\overline{\mu},\alpha$ we can numerically compute the threshold using \eqref{eq:condexp} and \eqref{eq:threshold}. Another approach is to compute $\widetilde N(\tau)$  by iteratively filling cache capacity so that the marginal hit rate improvements $\mean{\mu_m\big|N_m=\widetilde N(\tau),\tau_m=\tau}$ for each $\tau$ are approximately equal.\footnote{
This process involves splitting $[0,T]$ to small intervals and increasing the threshold at each interval one-by-one inspecting the marginal hit rate improvements.
Since the possible values of the threshold $\widetilde N(\tau)$ are discrete, we remark that a perfect equality cannot be achieved.}


Figure~\ref{fig: optimal caching plot} (left)  shows the  $\ABT$ threshold $\widetilde N(\tau)$ for different content age $\tau$; 
the dotted line corresponds to expected number of requests $\mu_m\tau_m$ for the content $m$ which is the $C^{\text{th}}$ most popular content in the active catalog.
Note that the optimal threshold roughly follows this line, although it is a bit higher for contents with small age. This indicates that the $\ABT$ policy differs from a simple frequency estimate since it is more conservative with recent contents, as their popularity estimates are less accurate. 
Figure~\ref{fig: optimal caching plot} (right) shows the density of the marginal hit rate improvement $\mean{\mu_m\big|N_m=\widetilde N(\tau),\tau_m=\tau}$ for each age $\tau$; the dotted line is the threshold $\theta(\gamma_c)$ which is the minimum conditional expected popularity of contents optimally stored in the cache and also the marginal hit rate improvement at which the iterative filling algorithm would stop.

There is an intuitive connection between $\ABT$ and the policy which solves optimally \eqref{eq:opt1} by caching the highest $\widehat \mu_m$ values, let us call it $\pi^*$. Similar to $\ABT$, we may think of $\pi^*$ as a threshold policy, only with a threshold 
which results from considering the $\gamma_c$-th upper-quantile of the empirical distribution of $\widehat \mu_m$, which is random and dependent on $(N_m),(\tau_m)$. 
Due to the differences between the two thresholds,  $\ABT$ decisions result in a few caching mistakes and thus in suboptimal hit rate performance. 
However, as the number of contents increases $\lambda\to\infty$, the random threshold of the optimal policy converges to that of $\ABT$. We establish this fact in the following Theorem:


\begin{theorem}[$\textsf{ABT}$ Optimality in  Many Contents Regime]\label{th:ABT}
For shot rate $\lambda$, consider two caching systems, one running the optimal policy $\pi^*(\lambda,T)$, and one with $\ABT$. Denote their average hit probabilities by $h^*(\lambda,T)$ and $h^{\ABT}(\lambda,T)$ respectively.
We let $\lambda$ and $C$ go to infinity together such that $\lim_{\lambda\to\infty}\frac{C}{\lambda T}=\gamma_c$. 
Then we have almost surely: \vspace{-0.09in}
\[
\lim_{\lambda\to\infty}\pi^*(\lambda,T)=\ABT,\vspace{-0.08in}
\]
in the sense that they asymptotically have the same threshold function, and thus they cache the same contents. 

Moreover, $\textsf{ABT}$ is almost surely asymptotically optimal:
\[
\lim_{\lambda\to\infty} h^{\ABT}(\lambda,T)=\lim_{\lambda\to\infty}h^*(\lambda,T)=h^*(\infty,T).\vspace{-0.1in}
\]
where \vspace{-0.08in}
$$h^*(\infty,T)\hspace{-0.04in}=\hspace{-0.04in}\frac{1}{\overline\mu T}\hspace{-0.04in}\int_{\tau}\hspace{-0.03in}\int_{Z_m}\hspace{-0.15in}\mu_m\prob{N_m\geq\widetilde N(\tau)|\mu_m,\tau_m=\tau}dZ_m\:d\tau$$
\end{theorem}
\begin{proof}
The proof is in appendix~\ref{app:a}.
\end{proof}

We call the regime $\lambda\to\infty$ the ``many contents'' regime. 
This is generally a reasonable regime in the caching context, where catalogs of contents typically contain millions, if not billions, of contents, and the caches are dimensioned so that they can store a fraction of the catalog of active contents.
A corollary of Theorem \ref{th:ABT} is that the threshold \eqref{eq:threshold} separates the seemingly most popular contents from the less popular ones, such that the fraction of contents deemed popular is exactly $\gamma_c=C/\lambda T$, which captures the relative cache size ($C/N$) in our model. Hence, we can think of this threshold as a way of separating the $C$ seemingly most popular contents from the rest. For the rest of the paper, we consider the many content regime and focus on the influence of the shot duration $T$; from now on, we omit the mention of $\lambda=\infty$ in $h^*(\infty,T)$.


\section{Aggregating Estimates From $L$ Caches}

We consider a hierarchy of $L$ caches connected to a central cache, as in  Figure~\ref{fig:distr}. 
 Each content request arrives first at one of the $L$ local caches, and then it is observed by the global cache. 
In this section we explain that the hit rate performance of the local caches can be improved if popularities are estimated at the global cache.



\subsection{Local vs Global Estimation}

\subsubsection{Request Model for Uncorrelated Requests}\label{sec:model2}

We clarify how the SNM model is used in the $L$ cache system. Let $N^{\mathcal L}_m(t)$ be an inhomogeneous Poisson process describing the requests for content $m$ which reach the global cache, and assume that $N^{\mathcal L}_m(t)$ is built using our SNM model from Section~\ref{sec:snm}, where the mean popularity of contents at the global cache is equal to $\overline\mu$. 
When a request is made, we randomly select one local cache uniformly with probability $1/L$ and assume that the specific local cache \spyros {was the one this request came from}. We denote by ${N}_m^l(t)$ the thinned inhomogeneous Poisson process observed at local cache $l$.

\subsubsection{Global is Faster}
As before we will fix a particular time instance ($t=0$) and study the behavior of our system at this instance. 
For this section, we will omit reference to absolute time, and track different times using the shot age $\tau$ with respect to the observation instance $t=0$.

We let $\Lcal$ be the set of local caches, with $|\Lcal|=L$. Denote by $N_m^l(\tau), \tau\geq 0$, the number of requests for content $m$ arrived at cache $l\in\Lcal$ in the time interval $[-\tau,0]$.\footnote{Here we slightly abuse the notation $N_m^l(\tau)$ to count requests in the interval $[-\tau,0]$ instead of $(-\infty,\tau]$.} Observe that $N_m^l(\tau_m)=N_m$ is the number of all requests for this content at cache $l$ so far, and $(N_m^l(\tau))_\tau$ corresponds to the entire history of requests for item $m$ and cache $l$. 
Note that $N_m^l(\tau)$ for $l=1,\dots,L$ are all independent Poisson processes with time-varying rate $\mu^l_m(\tau)$. In this section, we consider the case where the $\mu_m^l(\tau)$ are equal for all $l$; however, we keep the index $l$ to stress that it refers to a quantity at local cache $l$.

By the properties of thinning Poisson processes, the rate of the aggregate request process satisfies
\[
\mu^\mathcal L_m(\tau)=L{\mu}^l_m(\tau), \quad\forall m.
\]

To compare local versus global estimation we  define two ways in which a local cache $l$ can decide its caching vector $(y^l_m)_m$.

\noindent\textbf{Local estimation.} The caching policy $\pi^l$ takes as input the local request information only, i.e., it is a causal mapping of past local observations:
\[
(y_m^l)_m=\pi^l\left[(N_m^l(\tau))_{m,\tau}\right].
\]
Let $h_l^{\pi^l}(T)$ be the average hit probability of policy $\pi^l$ using local estimation when the SNM shots have duration $T$; here
$T$ has a profound impact on the quality of caching decisions since it determines the dynamicity of the model. For example, for a very small $T$, many contents are requested only for a very few times in their lifetime. The best local cache hit probability performance with local estimation is then \[h_l^*(T)=\max_{\pi^l}h_l^{\pi^l}(T).\]

\noindent\textbf{Global estimation.} The caching policy can take as input the collection of local request information, in this case we pass all requests as arguments and we write $\pi^\mathcal L$ instead of $\pi^l$
\[
(y_m^l)_m=\pi^\mathcal L\left[(N^l_m(\tau))_{m,l,\tau}\right],
\]
where the index $\mathcal L$ on $\pi^{\mathcal L}$ points out that the histories of all local caches are available to the policy.
Similarly as above, let $h_\mathcal L^{\pi^\mathcal L}(T)$ be the average hit probability of policy $\pi^\mathcal L$ using aggregate estimation.
The best hit probability performance with global estimation is \[h_\mathcal L^*(T)=\max_{\pi^\mathcal L}h_\mathcal L^{\pi^\mathcal L}(T).\]
 
The global estimation may use the observations from all locations $(N^l_m(\tau))_{m,l,\tau}$ to better detect changing popularities. This directly translates to a hit rate benefit, which we capture with the following result.

\begin{theorem}[Global is Faster]\label{th:agg}
Consider the SNM model of Section~\ref{sec:model2} whose  requests $(N^l_m(\tau))_{m,l,\tau}$ are observed by the global system, and a thinned version of them $(N^l_m(\tau))_{m,\tau}$ are observed by local cache $l\in\Lcal$. 
The maximum hit probability performance of global system $h_\mathcal L^{*}(T)$ compares to the performance of any local cache $h_l^*(T)$ in the following manner
\[
h_\mathcal L^{*}(T/L)=h_l^*(T),\quad \forall T>0.
\]
\end{theorem}
\begin{proof}
The proof is in appendix~\ref{app:b}. The proof shows more generally that \emph{for any} policy using local estimation for a system with shot duration $T$, we can define a policy using global estimation having the same performance for shot duration $T/L$, and conversely.
\end{proof}
According to Theorem \ref{th:agg} the global system aggregates more samples and its performance can be understood as virtually slowing down the popularity dynamics. 
Since faster dynamics have  a detrimental effect on hit rate, this virtual slowing down helps the global system to improve hit rate performance. 
Below we provide numerical performance comparison between local and global learning. We use the $\ABT$ policy in the many contents regime, where its optimality allows us to compute (in numerical terms) the exact benefit we have from aggregation. We define the hit probability gain as
\begin{align*}
G^{\ABT}(T)=h_\mathcal L^{*}(T)-h_l^{*}(T)=h_l^{*}(TL)-h_l^{*}(T).
\end{align*}
Figure~\ref{fig:lvg} plots $G^{\ABT}(T)$ and shows that gains reach $35\%$ of absolute hit rate improvement for a specific $T$. We observe that there is a wide range of values of $T$ for which the system greatly benefits from aggregating requests and learning faster.

\begin{figure}[t]
\begin{center}
   \begin{overpic}[scale=0.38]{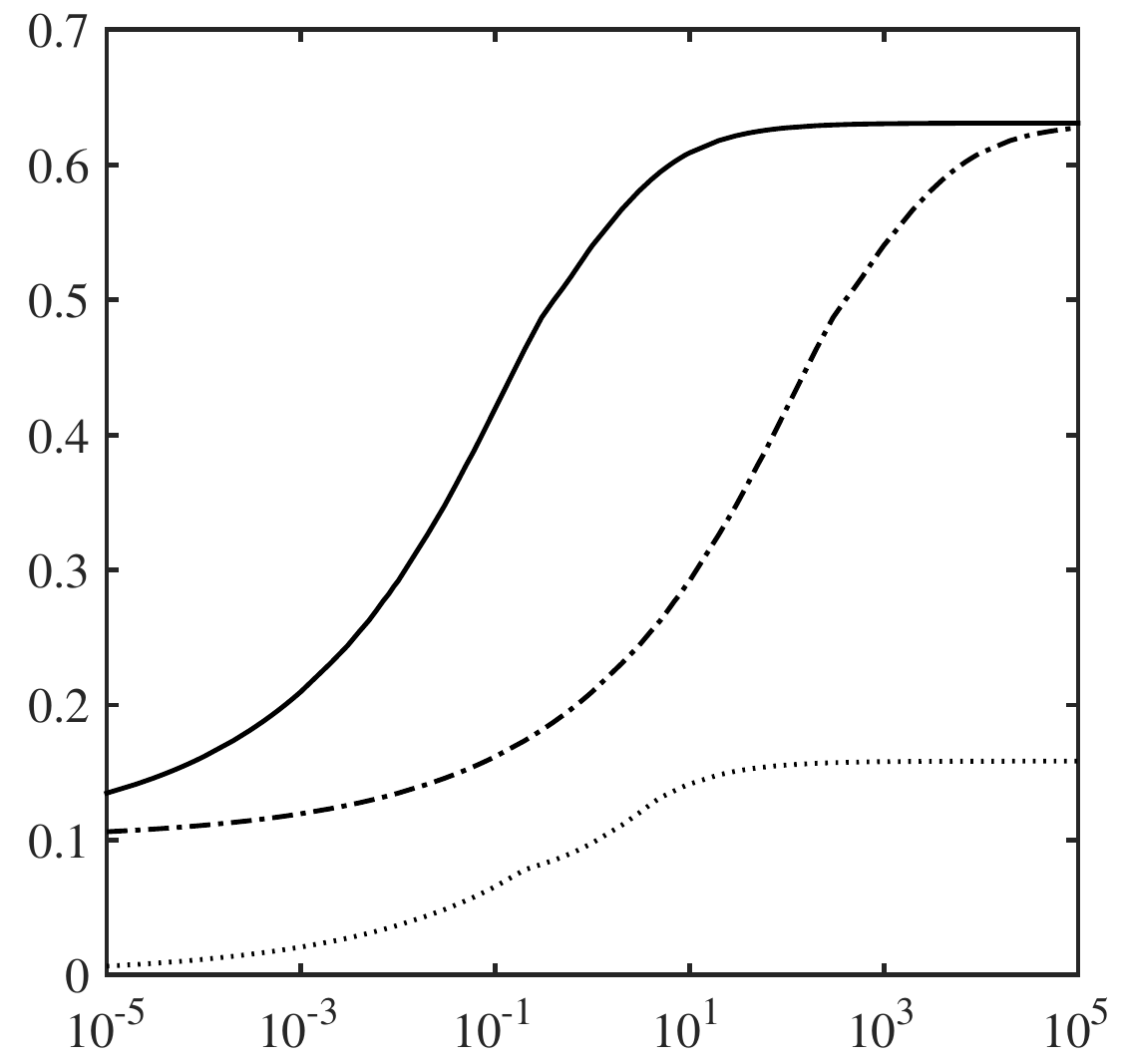}
        \put(32,-4){\scriptsize $T$ (shot duration)}
				\put(12,84){\scriptsize Average Hit Probability}
				\put(25,58){\tiny global }
				\put(23,52){\tiny $h^*_\mathcal L(T)$}
				\put(57.5,66){\tiny local}	
				\put(54.5,60){\tiny $h_l^*(T)$}	
				\put(75,28.5){\tiny whole files}	
      \end{overpic}
			\begin{overpic}[scale=0.38]{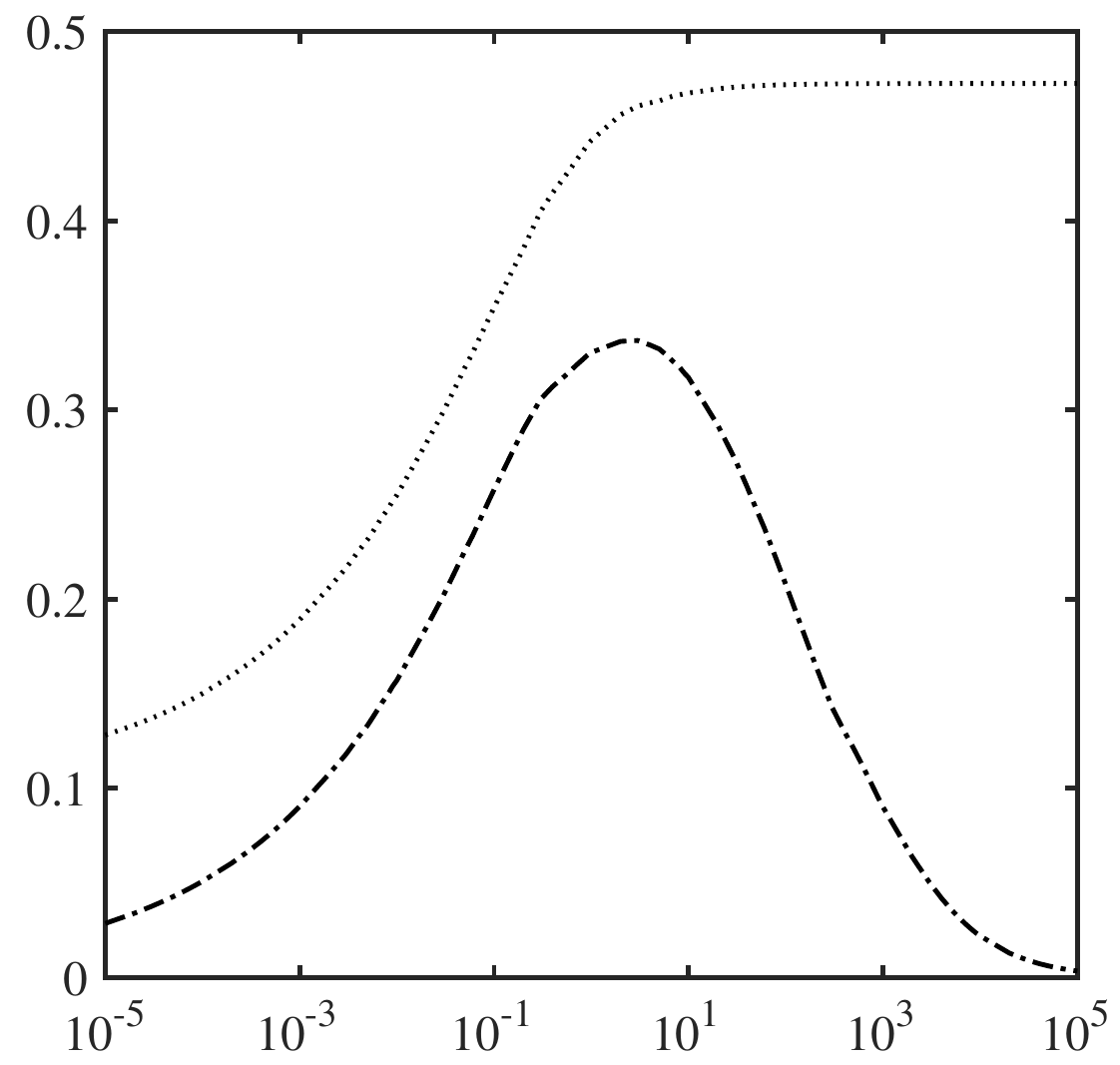}
        \put(32,-4){\scriptsize $T$ (shot duration)}
				\put(12,84){\scriptsize  Absolute Gain}
				\put(43,43){\tiny global vs local}
				\put(45,37){\tiny $G^{\ABT}(T)$}
				\put(59,82){\tiny global vs whole files}
      \end{overpic}
      \caption{\textbf{Global is faster for uncorrelated caches.} (left) Optimal  hit probability under time-varying popularity using (i) global learning,  (ii) local learning, or (iii) storing entire files (cache size reduced to $\gamma_C/\xi$). (right) Absolute hit probability gain. 
			\textbf{Parameters:} $\alpha=0.8$, $\overline\mu=20$, $\gamma_c=10\%$, $L=\xi=1000$.
			\vspace{-0.3in}}
            \label{fig:lvg}
            \end{center}
\end{figure}
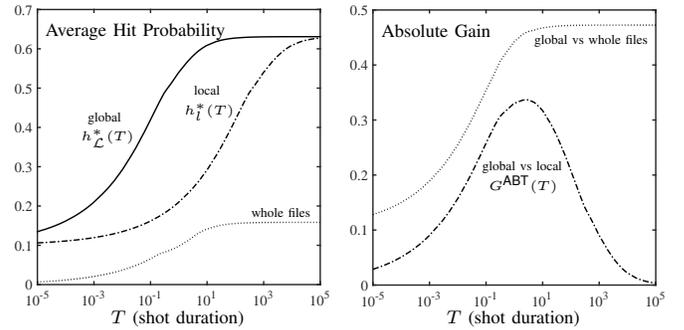

\subsection{Correlated Popularities}\label{sec:correlated_popularities}


We might expect that some popularities may vary from region to region; this could be attributed to different sociological and cultural backgrounds of users or different types of activities associated with these locations. 
For example, job commuters might pursue similar requests for contents
and hence office areas might ``see'' a particular request pattern. 
The caching benefit  from geographical locality of content has been recently pointed out \cite{kurose,scellato2011,huguenin2012}. 

In this section we assume that the contents exhibit geographical correlations and hence there exist groups of contents which are very popular in a subset of local caches. 
Although for identical local caches the best approach was to learn from the aggregation of  all caches, 
here it may be more efficient to restrain the aggregation to subsets of caches which witness similar traffic patterns.

\subsubsection{Request Model for Correlated Locations}

We propose here a model for correlated local popularities $(\mu_m^l)$. 
As before, aggregate popularities $(\mu_m^\mathcal L)$ are described by the SNM model of Section~\ref{sec:snm}, with global mean popularity $\mean{\mu_m^\mathcal L}=\overline\mu$.
To model correlations between the local cache popularities $(\mu_m^l)$, we draw inspiration from the field of community detection \cite{lelarge2013reconstruction} and inhomogeneous random graphs \cite{bollobas2007phase}:
\begin{itemize}
\item Each content $m$ is associated with a feature vector $X_m$. To simplify the model we let $(X_m)$ be independent uniform random variables taking values in $[0,1]$. 

\item Each location $l$ is associated with feature vector $Y_l$, which are again chosen independently and uniformly in $[0,1]$.

\item We define a kernel $K(x,y)=g(|x-y|)$, where $g$ is continuous, strictly decreasing on $[0,1/2]$, symmetric and $1$-periodic, with $\int_{[0,1]}g(|x-y|)dy=1$ for all $x\in[0,1]$. For such a correlation kernel $K$, we can think of the feature vectors $(X_m),(Y_l)$ as lying on the torus $[0,1]$ rather than the interval.
\end{itemize}
The local popularity of content $m$ at cache $l$ is defined as 
\begin{equation*}
\mu_m^l=\mu_m^\mathcal L\frac{K(X_m,Y_l)}{\sum_{l'\in\Lcal}K(X_m,Y_{l'})},\quad \forall m,l. 
\end{equation*}
As the number of caches $L$ increases and provided the kernel function $g$ satisfies some basic conditions, the normalization constant almost surely becomes deterministic: 
\[\frac1L\sum_{l'\in\Lcal}K(X_m,Y_{l'})\underset{L\to\infty}\to\int_{[0,1]}g(|X_m-y|)dy=1,\] 
so that 
\begin{equation}\label{eqn: limit local pop}
\mu_m^l\underset{L\to\infty}\to \mu_m^\mathcal L K(X_m,Y_l)/L.
\end{equation}

This basic model can easily be extended to multi-dimensional features, to capture more complex correlation structures.

\subsubsection{Local is More Accurate}
With correlated popularities, the popularity distribution is more skewed if observed on a subset of caches, and less skewed if aggregated over all caches. Since popularity skewness is advantageous to caching, it should not be surprising that learning in clusters can outperform learning globally.
Recall that 
${h_\mathcal L^*(T)}$ denotes the maximum average hit probability achieved by observing the aggregated request, and ${h_l^*}(T)$ the corresponding maximum average hit probability when observing local requests at cache $l$. 
Due to convexity, we have the following result.
 
\begin{theorem}[Local is More Accurate - Known Popularities]\label{th:corr}
In the limit of a static system (i.e., assuming $T\to\infty$), the hit probability performance of local learning is higher than that of aggregate global learning, i.e., for any global popularity distribution $(\mu_m^\mathcal L)$ it holds that
\[
{h_l^*}(\infty)\geq {h_\mathcal L^*}(\infty).
\]

Furthermore, 
 as the number of edge caches $L$ tends to infinity (such that \eqref{eqn: limit local pop} holds), the maximum expected hit probability ${h_l^*}(\infty)$ is almost surely:
\begin{align*}
\lim_{L\to\infty}{h_l^*}(\infty)&=\frac{1}{\overline\mu}\mean{\mu_m^l\ind(\mu_m^l\geq\theta^l)}\\
&=\int\bigg\{2\int_0^{g^{-1}\left(\frac{L\theta^l}{\mu_m^\mathcal L}\right)} g(t) \:dt\bigg\}\frac{\mu_m^\mathcal L}L\:dZ_m,
\end{align*}
where $g^{-1}(t)=\min\{x\geq0:\:g(x)\geq t\}\in[0,1/2]$ is the inverse of $g$ and $\theta^l$ is the unique value satisfying
\[
\gamma_c=\prob{\mu_m^l\geq\theta^l}=2\int g^{-1}\left(\frac{L\theta^l}{\mu_m^\mathcal L}\right)\:dZ_m.
\]
\end{theorem}
\begin{proof}
The proof is in appendix~\ref{app:c}.
\end{proof}


The above theorem characterizes the performance of correlated caches for nearly-static popularities.
The benefit in this case is due to the skewness of the local popularity distribution. 
However, when the popularities are unknown, we saw that aggregation is beneficial. In fact we observe a tradeoff; (i) aggregating all observations improves performance by collecting more samples and having more accurate estimates, but (ii) aggregating in subsets allows for more accurate popularity models with higher skewness value. 
\subsubsection{Clustering}
In order to retain the benefit from the increased skewness of the local popularities and at the same time  capture faster dynamics than local estimation would allow, we need to estimate the local popularities of contents based on the global information. One way to achieve this goal is to identify locations with similar local popularity profiles and to aggregate samples from these locations only. Leveraging our correlated local popularities model and the understanding of optimal policies gained from the previous sections, we can explore the tradeoff between capturing local popularities and detecting faster dynamics.

For this paper, we leave aside the problem of determining which local caches should be clustered together from the requests history. Instead, we assume that we know the embeddings of local caches $(Y_l)$. When the kernel $K(x,y)=g(|x-y|)$ has the simple form assumed here, the local caches which should be aggregated are those with similar feature vectors $Y_l$. 
We will consider a cluster $\mathcal S\subseteq\mathcal L$ covering the subset $[0,\omega]$ of the feature space, and study feature vectors $(Y_l)_{l\in\mathcal S}$.
For $\omega=1/k$, $k\in\N$, this is equivalent to considering $k$ disjoint clusters covering equal portions of the feature space.
 A meaningful regime here is to let the total number of local caches $L$ increase, i.e., $L\to\infty$, which means that we are looking at smaller and smaller local user populations, while keeping the global popularity distribution fixed, i.e., $\overline\mu$ and $T$ are fixed. As $L\to\infty$, we have $\frac{|\mathcal S|}{L}\to\omega$, and the aggregated popularity of items $m$ within $\mathcal S$ equals
\[
\mu_m^\mathcal S=\sum_{l\in\mathcal S}\mu_m^l=\frac{\mu_m^\mathcal L}{L}\sum_{l\in\mathcal S}K(X_m,Y_l)\hspace{-0.03in}\underset{L\to\infty}\to\hspace{-0.03in}\mu_m^\mathcal L\hspace{-0.03in}\int_0^\omega\hspace{-0.1in}g(|X_m-y|)dy.
\]
As discussed in previous sections, the aggregated popularity $\mu_m^\mathcal S$ is not known and has to be estimated from the aggregate requests $N_m^\mathcal S=\sum_{l\in\mathcal S}N_m^l$. We denote by $h_\mathcal S^*(T)$ the limit as $L\to\infty$ of the optimal hit probability averaged over local caches in the cluster $\mathcal S$ for shot duration $T$, when popularities are estimated from $N_m^\mathcal S$; we call this quantity the clustered hit probability. Leveraging the techniques from the previous sections, the clustered hit probability is obtained by appropriately defining an aging-based threshold $\widetilde N_m^\mathcal S(\tau)$ to (approximately) equalize the marginal improvements $\mean{\mu_m^\mathcal S\big|N_m^\mathcal S=\widetilde N^\mathcal S(\tau),\tau_m=\tau}$ for all $\tau\in[0,T]$. This yields
\begin{align*}
&h^*_\mathcal S(T)=\frac1{\overline\mu\omega}\mean{\mu_m^\mathcal S\ind\left(N_m^\mathcal S\geq\widetilde N_m^\mathcal S(\tau_m)\right)}=\frac{1}{\overline\mu\omega T}\\
&\hspace{-0.05in}\times\hspace{-0.05in}\int_\tau\hspace{-0.03in}\int_{Z_m}\hspace{-0.03in}\int_{X_m}\hspace{-0.15in}\mu_m^\mathcal S\prob{\hspace{-0.02in}N_m^\mathcal S\geq\widetilde N^\mathcal S(\tau)\big|\mu_m^\mathcal S,\tau_m=\tau\hspace{-0.03in}}\hspace{-0.02in}dX_mdZ_md\tau,
\end{align*}
where the thresholds $\widetilde N^\mathcal S(\tau)$ also ensure the correct fraction of contents is stored:
\begin{align*}
\gamma_c&=\prob{N_m^\mathcal S\geq\widetilde N_m^\mathcal S(\tau_m)}\\
&=\frac{1}{T}\int_\tau\hspace{-0.03in}\int_{Z_m}\hspace{-0.03in}\int_{X_m}\hspace{-0.15in}\prob{N_m^\mathcal S\geq\widetilde N^\mathcal S(\tau)\big|\mu_m^\mathcal S,\tau_m=\tau}dX_mdZ_md\tau.
\end{align*}
Figure~\ref{fig: performance curves}-(left) shows the clustered hit probability $h_\mathcal S^*(T)$ as a function of the shot duration $T$ for different sizes $\omega$ of the cluster and for a particular kernel $g(x)=5\left(1-2x\right)^4$.
As $T\to\infty$, smaller values of $\omega$ (i.e. smaller clusters) yield better hit rates, as stated in Theorem~\ref{th:corr}; however, as the system becomes more dynamic (i.e., for smaller $T$) small cluster sizes fail to estimate the aggregated popularities $\mu_m^\mathcal S$, which results in poor clustered hit rate.
Cluster size $\omega=1$ corresponds to learning by aggregation of all requests in one big cluster, which is advantageous when $T$ is small. 
 The right panel shows the kernel $g(x)$, along with smoothed versions of the kernel $\int_0^\omega g(|x-y|)dy$ corresponding to using a cluster size $\omega$. 
When $\omega=1$ (global learning) the corresponding smoothed kernel is flat, resulting in loss of location information. 
However, smaller values of $\omega$ yield smoothed kernels which approximate better and better the true correlation kernel, which allows to capture the local popularity characteristics.


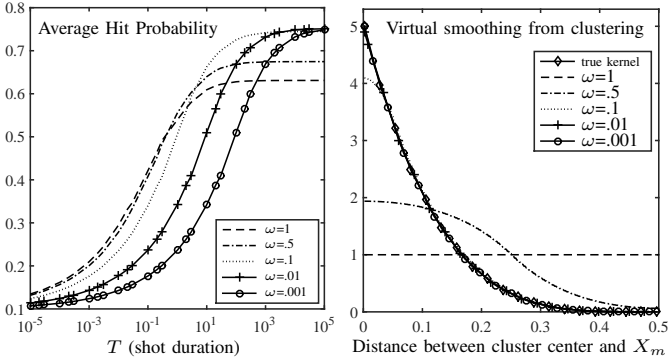
\begin{figure}[t]
\begin{center}
   \begin{overpic}[scale=0.415]{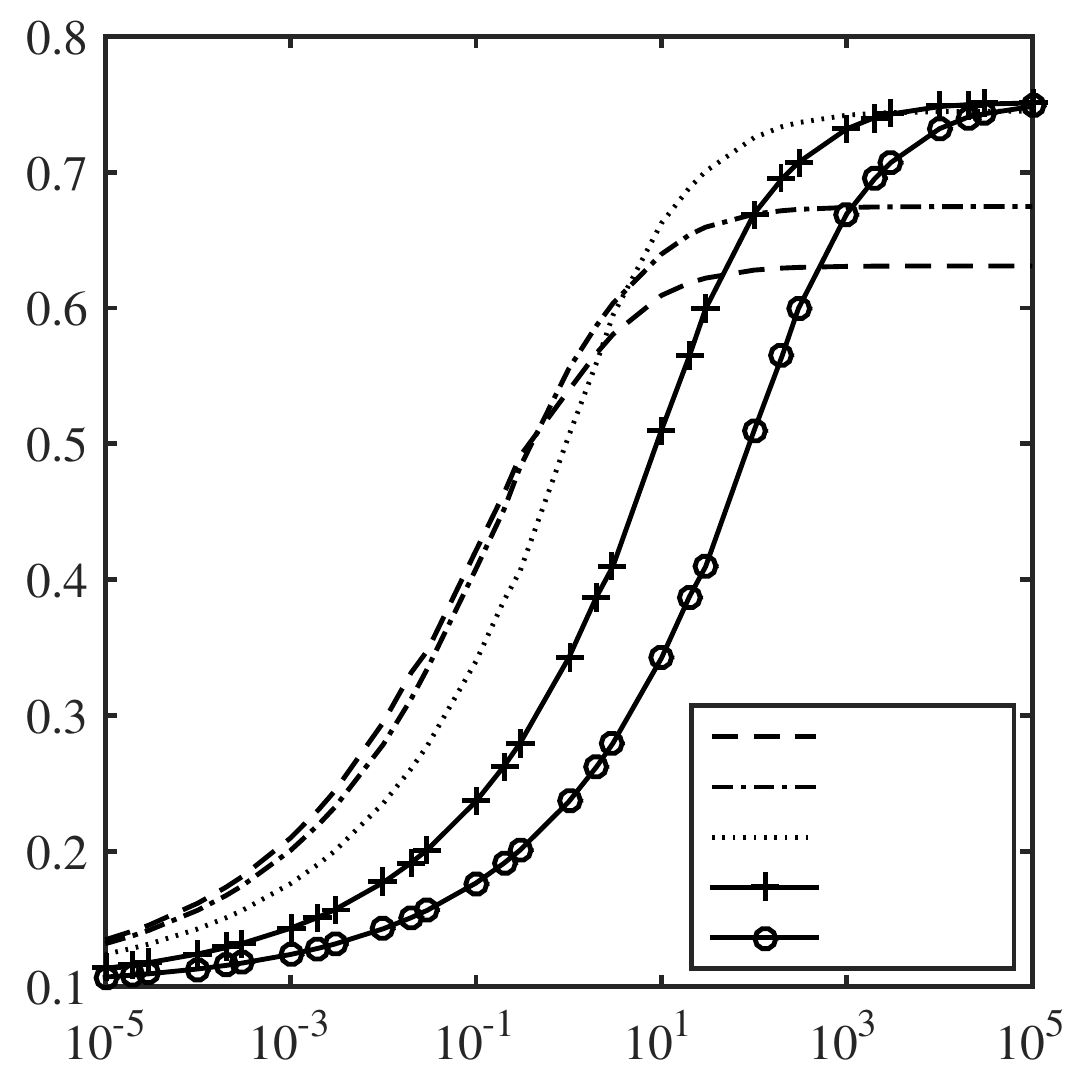}
        \put(32,-4){\scriptsize $T$ (shot duration)}
				\put(12,89){\scriptsize Average Hit Probability}
				\put(78,31){\tiny $\omega$=1}
			  \put(78,26.4){\tiny $\omega$=.5}
				\put(78,22){\tiny $\omega$=.1}
				\put(78,17.2){\tiny $\omega$=.01}
				\put(78,12.4){\tiny $\omega$=.001}
      \end{overpic}
			\begin{overpic}[scale=0.417]{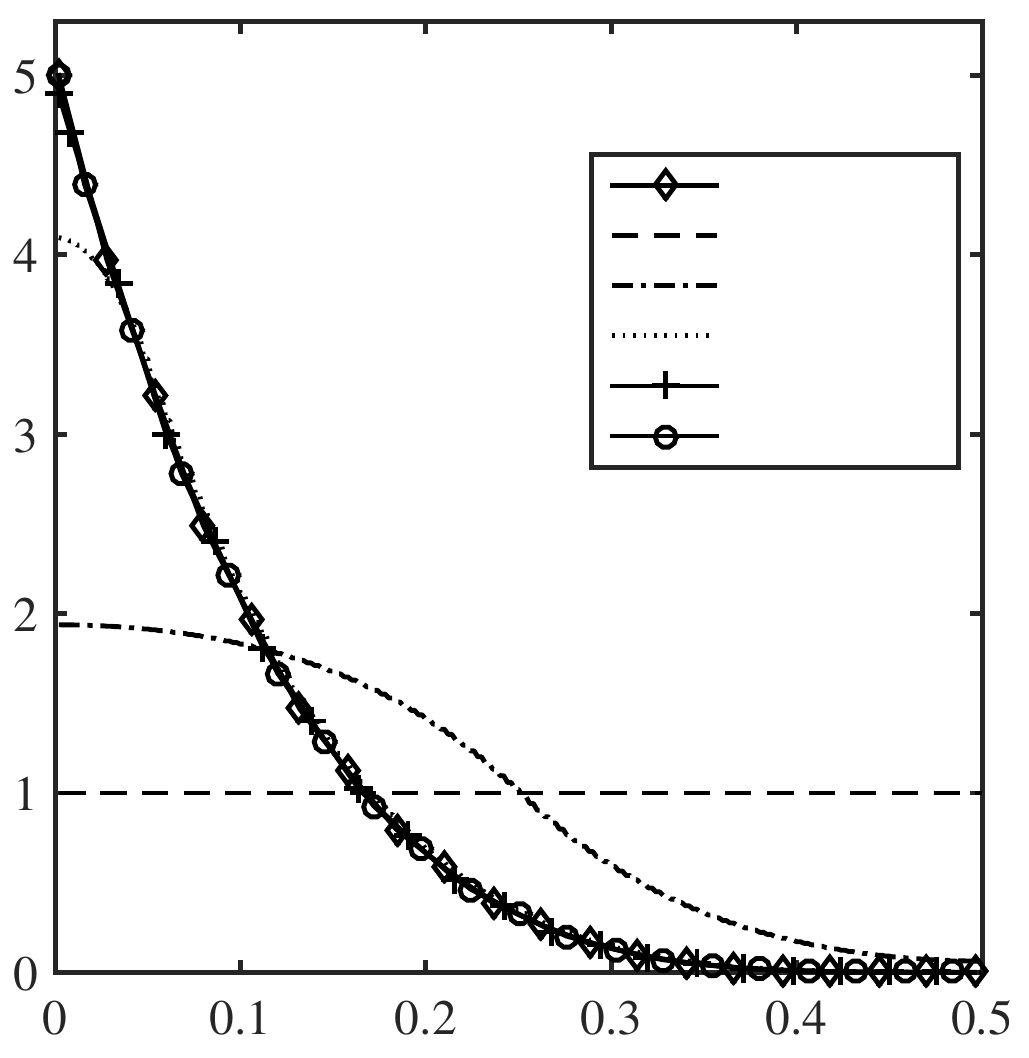}
        \put(2,-4){\scriptsize Distance between cluster center and $X_m$}
				\put(12,90.5){\scriptsize  Virtual smoothing from clustering}
				\put(70,81.2){\tiny true kernel}
				\put(70,76.4){\scriptsize $\omega$=1}
				\put(70,71.6){\scriptsize $\omega$=.5}
				\put(70,66.8){\scriptsize $\omega$=.1}
				\put(70,62){\scriptsize $\omega$=.01}
				\put(70,57.2){\scriptsize $\omega$=.001}
      \end{overpic}
      \caption{\textbf{Learning in clusters for correlated caches.} (left) Optimal hit probability under time-varying popularity using clustering for different relative size $\omega$ of cluster. (right) True kernel vs. smoothed kernels resulting from clustering. 
			\textbf{Parameters:} $\alpha=0.8$, $\overline\mu=1$, $\gamma_c=10\%$.
			\vspace{-0.25in}}
            \label{fig: performance curves}
            \end{center}
\end{figure}

\section{To Prefetch Or Not}\label{sec:prefetch}

So far our analysis did not consider the traffic required for placing the content in the local caches. In this section we
 focus on practical policies which either perform adaptive caching (without prefetching) or explicitly prefetch content which is not yet requested. 
We first introduce popularity scores which are calculated by the global cache and then made available at the local caches. 
These scores can be used both for performing score-gated LRU as well as for determining which contents to prefetch. 
We present simulations of the proposed techniques and showcase that prefetching is of fundamental significance to small population caches.

\subsection{Age-based Popularity Scores}

As explained in the previous section, it is advantageous to estimate popularities at the global cache and then use the estimates at the local caches. One standard methodology to coordinate this mechanism is to use \emph{content scores}. A score is simply a value per content which can be used to perform caching. For example, we may give value 1 to very popular contents and value 0 to the rest. 

We exploit intuition from the one-cache analysis in section~\ref{sec:onecache} to propose the use of threshold functions $\widetilde N(\tau;\gamma_c)$ as scores. 
Recall the definition of the threshold $\widetilde N(\tau;\gamma_c)$
of $\ABT$ policy
\begin{equation*}
\widetilde N(\tau;\gamma_c)=\min\{k\in \N : \mean{\mu_m|N_m=k,\tau_m=\tau}\geq \theta(\gamma_c)\},
\end{equation*}
where $\theta(\gamma_c)=F^{-1}_{\widehat\mu_m}(1-\gamma_c)$.
Then, the contents which satisfy $N_m>\widetilde N(\tau_m;\gamma_c)$ are the $\gamma_c$ with highest popularity estimates and should be cached; 
here $\gamma_c$ is the equivalent cache size.
We may produce new thresholds by choosing different values for $\gamma_c$. In particular let us pick $\beta_1$, $\beta_2$ such that
\[
1\geq \beta_1> \gamma_c >\beta_2\geq 0.
\]
Replacing $\gamma_c$ with $\beta_1$ or $\beta_2$ is equivalent to considering the $\ABT$ policy on a virtual cache with larger or smaller cache size respectively. In particular, if we use $\beta_1>>\gamma_c$ the virtual cache is larger, hence the threshold $\widetilde N(\tau;\beta_1)$ smaller: almost all contents will pass the threshold. Clearly we cannot store all these in our cache (which is of size $\gamma_C$), but if a content with $(N_m,\tau_m)$ does not satisfy $N_m\geq \widetilde N(\tau_m;\beta_1)$ we may infer that it is ``super unpopular'', see Figure~\ref{fig:multifig}-(a). Similarly, by using $\beta_2<<\gamma_c$ hence a small virtual cache, only the ``super popular'' contents will satisfy the threshold.
The idea is to use  the function $S_m(\beta) =\ind\left(N_m>\widetilde N(\tau_m;\beta)\right)$ as a generalized score for the popularity of content $m$. 
What is convenient in this definition is that we take into account the age of content without complicating the design of scores.


\begin{figure*}[t]
	\centering
	\hspace{0.03in}
	\begin{overpic}[scale=0.15]{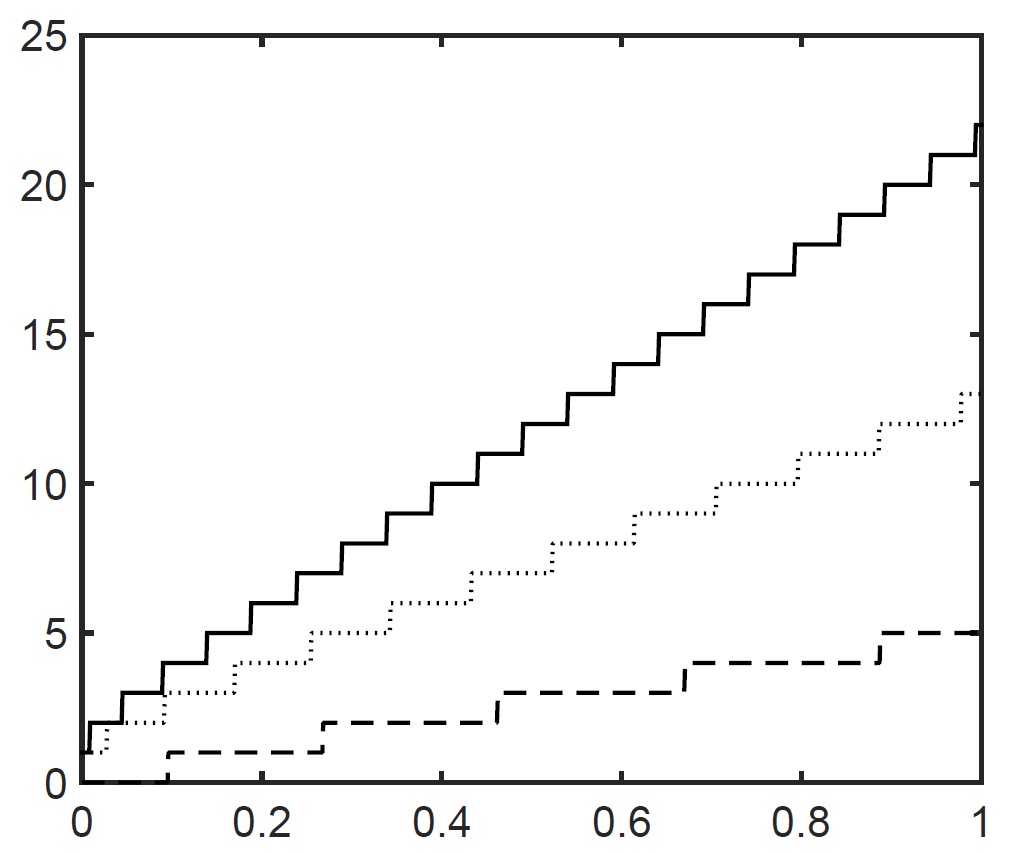}
	\put(5,-6){\small (a)}
		\put(75,12){\footnotesize gate}
		\put(25,54){\footnotesize prefetch}
		\put(41,-6){\footnotesize Age $\tau$}
		\put(-5,15){\footnotesize \rotatebox{90}{Number of Requests}}
		\put(60,41.5){\scriptsize $\widetilde N(\tau ;\gamma_c)$}
		\put(60,22){\scriptsize  $\widetilde N(\tau ;\beta_1)$}
		\put(60,65){\scriptsize  $\widetilde N(\tau ;\beta_2)$}
	\end{overpic}
		\hspace{0.16in}
	\begin{overpic}[scale=0.17]{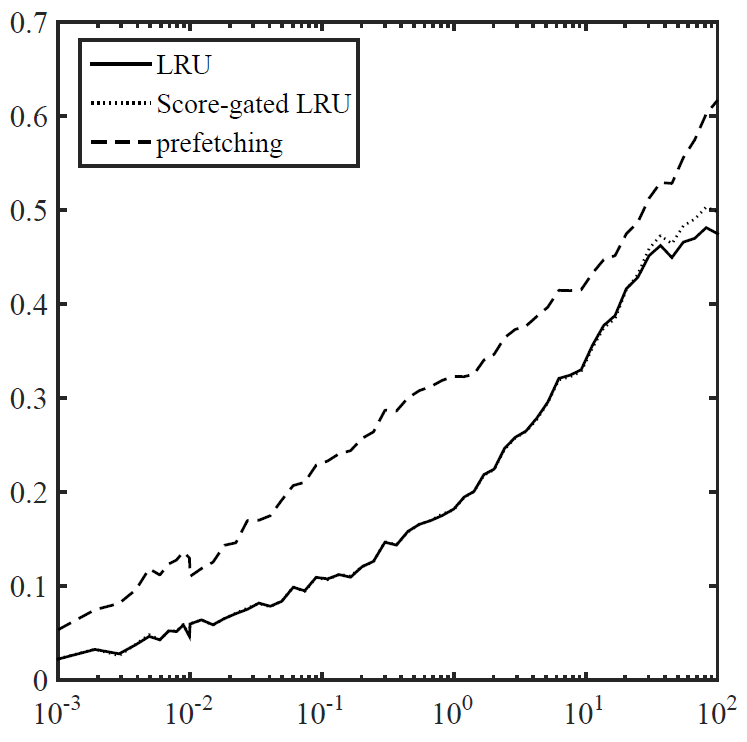}
	\put(5,-6){\small (b)}
		\put(-6,35){\footnotesize \rotatebox{90}{Hit Probability}}
				\put(25,-6){\footnotesize Shot duration $T$}
	\end{overpic}
	\hspace{0.16in}
	\begin{overpic}[scale=0.17]{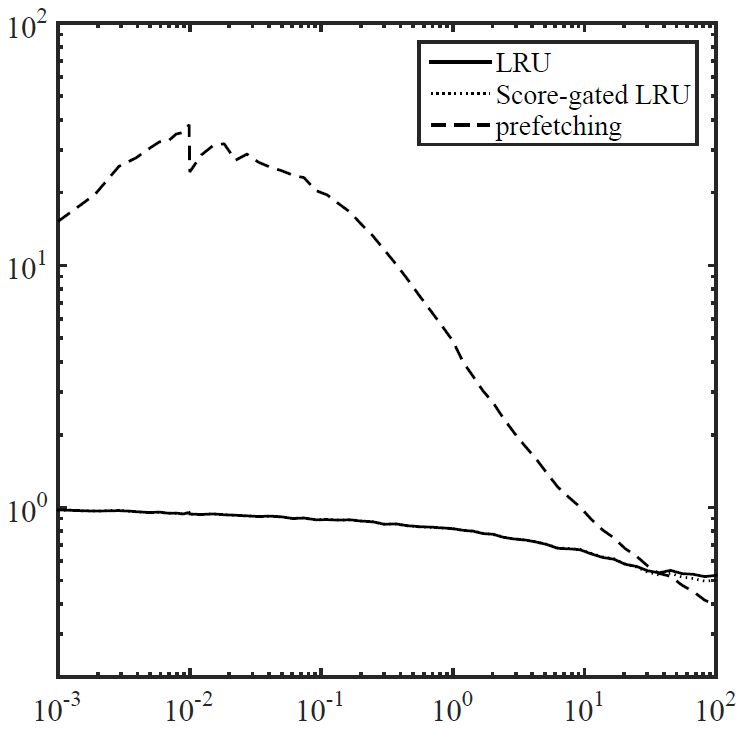}
	\put(5,-6){\small (c)}
		\put(-6,7){\footnotesize \rotatebox{90}{Transmissions per request}}
		\put(25,-6){\footnotesize Shot duration $T$}
	\end{overpic}
	\hspace{0.16in}
	\begin{overpic}[scale=0.23]{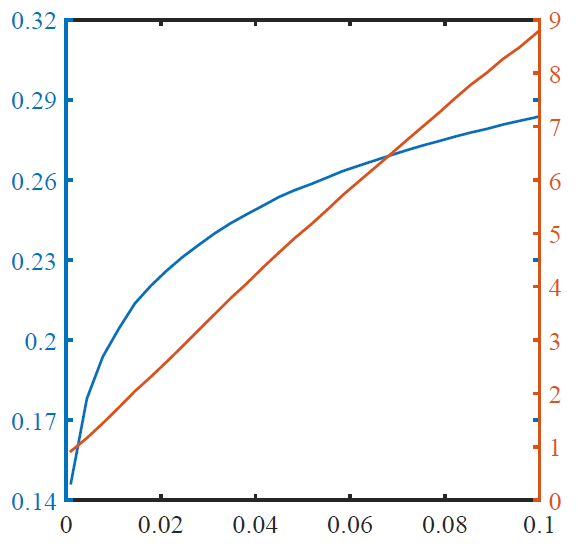}
	\put(5,-6){\small (d)}
		\put(31,-6){\footnotesize Threshold $\beta_2$}
		\put(-6,35){\footnotesize \rotatebox{90}{Hit Probability}}
		\put(98,94){\footnotesize \rotatebox{-90}{Transmissions per request}}
	\end{overpic}
	\caption{Simulation of LRU policies with scores. (a) Thresholds for age-based scores. 
	(b) Hit probability performance comparison. (c) Traffic footprint performance comparison. (d) Performance of LRU with prefetching vs threshold parameter $\beta_2$, hit probability (read left) and transmissions per request (read right). \textbf{Simulation Parameters:} $\alpha=0.8$, $\lambda=10^{4}/T$, $\overline\mu=10$,  $\beta_1=50\%$, $\gamma_C=10\%$, $\beta_2=5\%$, $L=\xi=1000$.}\vspace{-0.2in}
	\label{fig:multifig}
\end{figure*}
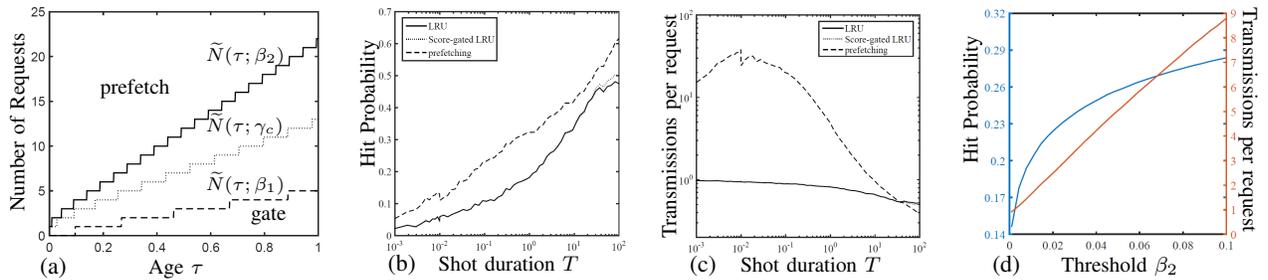

\subsection{Score-gated LRU}

The Least Recently Used (LRU) replacement rule is one of the most widely used caching policies. 
An intuitive way to implement LRU is to maintain a linked-list where the contents are always stored from the most recently used (head) to the least recently used (tail). A new request puts the new content at the head and pushes all contents by one position in the list causing the eviction of the content from the tail. Requests for existing contents simply bring the content to the head.

LRU is desirable in practice because it is purely adaptive
and simple to implement. 
Nevertheless it performs quite poorly in our setup. A traditional improvement over LRU is the so-called \emph{score-gated LRU}, whereby the content requests are filtered using a threshold on the content score. The high-score requests follow the LRU rule, while the low-score ones are never cached. 
In the context of our model we may use the age-based scores  $S_m(\beta_1)$ to  
perform score-gated LRU using the function . 

In Figure~\ref{fig:multifig}-(b),(c) we compare LRU and gated-LRU for different shot durations $T$, in a hierarchy of $L=1000$ caches. 
We use $\beta_1=50\%$, which means that only $50\%$ of highly unpopular content is not cached.
In particular, Figure~\ref{fig:multifig}-(c) shows that the number of transmissions from the origin server are roughly the same for the two policies, since for adaptive policies this is equal to one minus the hit probability.
Figure~\ref{fig:multifig}-(b) shows that gated-LRU benefits from predictions only at high values of $T$ where the popularity is semi-static. In the more dynamic scenarios each one of the 1000 local caches only receives a handful of requests per content and hence no adaptive caching policy can be effective.

\subsection{LRU with Prefetching}

With adaptive policies the first local request for a content is always a miss, a fact which may hurt hit rates in dynamic settings. 
To amend this situation we propose prefetching content to local caches whenever it is deemed ``super popular'' by the global cache.
There is a simple coordination mechanism for this.
In a periodic fashion (here we choose a period $T$), the controller sends a fake request for content $m$ if $S_m(\beta_2)=1$. If a local cache does not have content $m$, it performs a prefetching operation, else it simply brings the content to the head.
Figure~\ref{fig:multifig}-(b) shows that the performance of this proactive mechanism successfully improves the hit rate.
	
	Although prefetching is clearly improving hit rates, 
	we need to measure the amount of traffic it induces.
Figure~\ref{fig:multifig}-(c) shows how many times on average each content is transmitted over the backhaul for each request.
In very dynamic settings we may need to make 30 transmissions in total to deliver one request. Although this might seem expensive, if we factor in the considerations of partial caching and the coefficient $\xi$, we conclude that this corresponds roughly to $30/\xi$ increase in total used bandwidth, which in this example is equal to $3\%$, which is extremely low for a system with $1000$ cache. 
Figure~\ref{fig:multifig}-(d) shows how hit rate and traffic footprint tradeoff for $T=1$ when we vary threshold $\beta_2$; the diminishing returns in hit rate may motivate the use of smaller values of $\beta_2$ in a joint consideration of hit rate and traffic footprint performance.

\section{Conclusions}
Our work focuses on learning time-varying popularities at wireless access caching.
An architecture which combines global learning and local caches with small population is proposed to improve the latency of accessing wireless content.
It is shown that age-based thresholds can timely exploit time-varying popularities to improve caching performance. Moreover, the caching efficiency is maximized by a combination of global learning and clustering of access locations. Score mechanisms are then proposed to help with practical considerations at local caches.

\bibliography{IEEEabrv,Caching}
\bibliographystyle{IEEEtran}

\appendices

\section{Proof of Theorem \ref{th:ABT}}\label{app:a}

\begin{proof}
First we reformulate the  optimal policy for a finite $\lambda$ as a threshold policy. Then, we show that in the limit as $\lambda\to\infty$ the threshold function of the optimal policy becomes deterministic and equal to that of the $\ABT$, which also implies that the $\ABT$ policy is asymptotically optimal.

Recall that the optimal solution of the optimization problem~\eqref{eq:opt1} is to store the $C$ items with the highest values of $\widehat \mu_m$. For each item $m$, $\widehat \mu_m$ is independent of the other items; let $F$ be the distribution of $\widehat \mu_m=\mean{\mu_m|N_m,\tau_m}$ and let $F^\lambda(x)=\frac1{|\Mcal|}\sum_{m\in\Mcal}\delta_{\widehat \mu_m}(x)$ be the empirical distribution of the $\widehat \mu_m$'s, where $\delta_x$ is a Dirac function at $x$. The distribution $F$ has a density, because, given any value of $N_m$, $\widehat \mu_m$ is a smooth decreasing function of the continuous random variable $\tau_m$. As a consequence, for finite $\lambda$, all the values $(\widehat \mu_m)$ are almost surely distinct. Thus, a more intricate but equivalent way to define the optimal policy is that it stores all the items $m$ which have $\widehat \mu_m$ larger or equal to a threshold $\theta^\lambda(C)$, where we set $\theta^\lambda(C)$ equal to the $\frac C{|\Mcal|}$-th upper-quantile of the empirical distribution $F^\lambda$, i.e., $\theta^\lambda(C)$ is the largest value such that $\PP_{F^\lambda}\left(\widehat \mu_m\geq\theta^\lambda(C)\right)=\frac C{|\Mcal|}$. Note that, if $F$ did not have a density, it would not always be possible to find such a value $\theta^\lambda(C)$, which exactly separates the contents with the top $C$ estimates $\widehat \mu_m$ from the rest. Indeed, if $F$ had atoms, many contents could have the same value of $\widehat \mu_m$; there would then be a need for a tie-breaking rule to decide between content with the same estimate $\widehat \mu_m$.

We now let the shot arrival rate and the cache size tend to infinity together, i.e., $\lambda\to\infty$ with $\lim_{\lambda\to\infty}\frac{C}{\lambda T}=\gamma_c$. The size $|\Mcal|$ of the set of active contents is a Poisson random variable with mean $\lambda T$, so $\frac C{|\Mcal|}\to\gamma_c$. Also, the $\widehat \mu_m$'s are independent samples from $F$, so their empirical distribution tends to $F$, i.e., $F^\lambda\to F$ almost surely. Then, the $\frac C{|\Mcal|}$-th upper-quantile $\theta^\lambda(C)$ of $F^\lambda$ tends to the $\gamma_c$-th upper-quantile $\theta(\gamma_c)$ of $F$, which is well-defined because $F$ has a density. For each value of $\tau$, we let $\widetilde N(\tau)$ be the smallest integer $k$ such that $\mean{\mu_m\big|N_m=\widetilde N(\tau),\tau_m=k}\geq\theta(\gamma_c)$; this is the age-dependent threshold of the $\ABT$ policy.
Then, in the many-content regime, we have almost surely that $\widehat \mu_m\geq\theta(\gamma_c)$ if and only if $N_m\geq\widetilde N(\tau_m)$, which means the $\ABT$ policy is the limit as $\lambda\to\infty$ of the finite-system optimal policy $\pi^*(\lambda,T)$. Finally, $\PP_{F^\lambda}\left(\widehat \mu_m\geq\theta(\gamma_c)\right)\underset{\lambda\to\infty}\to \PP_F\left(\widehat \mu_m\geq\theta(\gamma_c)\right)=\gamma_c$, which means the fraction of contents initially stored by the $\ABT$ policy tends to $\gamma_c$, and thus the last stage of $\ABT$ which sets arbitrary $y_m$'s to 0 to ensure the cache size constraint affects a negligible number of contents. This implies $\lim_{\lambda\to\infty} \left|h^{\ABT}(\lambda,T)-h^*(\lambda,T)\right|=0$; hence, the $\ABT$ policy is asymptotically optimal. The expression for the asymptotic optimal expected hit probability $h^*(\infty,T)$ follows directly from the expression of the $\ABT$ policy.
\end{proof}

\section{Proof of Theorem \ref{th:agg}}\label{app:b}
\begin{proof}
 To prove the result we will make a connection between three arrival processes, (i) the observations at a local cache, (ii) the observations at the global cache, and (iii) the observations at the global cache about a system with accelerated time. In particular we will show that (i) and (iii) are the same in distribution. Then, we will show that the corresponding caching mappings are the same, from which the hit rate comparison will follow.

We will need an intermediate step. Define the \emph{$L$-speedup dynamics} $\nu_m^\mathcal L(\tau)=\mu_m^\mathcal L(L\tau)$, for which time evolves $L$-times faster than in the original system.
Under an $L$-speedup, a SNM process with rate $\lambda$ becomes a SNM with rate $\lambda L$ and the shot duration is shrunk to $T/L$. The associated Poisson process $S^\mathcal L_m(\tau)$ of requests for content $m$, with rate $\nu_m^\mathcal L(\tau)$, is called the $L$-speedup aggregate requests process. As the next lemma states, the $L$-speedup aggregate requests process is statistically identical to the request process which an individual local cache receives in the original system.
\begin{lemma}[Speedup Statistics]\label{lem:agg}
\[
S^\mathcal L_m(\tau)\stackrel{d}{=}N_m^l(\tau),\quad \forall l=1,\dots,L, \forall \tau, \forall m.
\]
\end{lemma}
\begin{proof}
Let us consider any interval $I=(a,b)$ in $\R_+$. We can define $N_m^l(I)=N_m^l(b)-N_m^l(a)$ and $\mu_m^l(I)=\int_I\mu_m^l(\tau)d\tau$; this simply means we consider the number of requests occurring at cache $l$ during the time interval $I$ rather than from a fixed time until time $0$. $N_m^l(I)$ is a Poisson random variable with mean $\int_I\mu_m^l(\tau)d\tau$. Similarly, $S^\mathcal L_m(I/L)$ is a Poisson random variable with mean $\int_{I/L}\nu_m^\mathcal L(\tau)d\tau$. In addition, we have $$\int_{I/L}\nu_m^\mathcal L(\tau)d\tau=\frac{1}{L}\int_I\mu_m^\mathcal L(\tau)d\tau=\int_I\mu_m^l(\tau)d\tau,$$ so that $S_m^\mathcal L(I/L)\overset d=N_m^l(I)$ for all intervals $I\subseteq \R_+$, which shows the two process are statistically identical.
\end{proof}
Using Lemma~\ref{lem:agg} we can establish a one-to-one mapping between policies $\pi$ using local information and policies $\pi^\mathcal L$ using global information for the $L$-speedup dynamics. Let $z_m^{\pi^\mathcal L}$ be the vector indicating which contents are stored at time $0$ for the $L$-speedup dynamics under policy $\pi^\mathcal L$, i.e., $z_m^{\pi^\mathcal L}=\pi_m^\mathcal L[(S_m^\mathcal L(\tau))_{m,\tau}]$. For any policy $\pi^l$ using local information and for any realization of the requests processes, we can define $\pi^\mathcal L$ as $\pi_m^\mathcal L[(S_m^\mathcal L(\tau))_{m,\tau}]=\pi_m^l[(N_m^l(\tau))_{m,\tau}]$ for all $m$. Using Lemma~\ref{lem:agg}, the hit probabilities are the equal in distribution under $\pi^l$ and for the $L$-speedup dynamics under $\pi^\mathcal L$. Hence, $h_\mathcal L^{\pi^\mathcal L}(T/L)=h_l^{\pi^l}(T),\:\forall T>0$. The same reasoning can be made starting from any policy $\pi^\mathcal L$ using aggregate information to define a policy $\pi^l$ using only the local information of cache $l$ for the original system such that the two policies have identical performance in distribution. The theorem then follows immediately by considering optimal policies in both directions.
\end{proof}

\section{Proof of Theorem \ref{th:corr}}\label{app:c}
\begin{proof}
By definition, we have $\sum_{l=1}^L\mu_m^l=\mu_m^\mathcal L$ for all $m$. The local popularity distributions $\mu_m^l$ are identically distributed for each edge cache $l$, so that $$L\mean{\mu_m^l\big|Z_m}=\mu_m^\mathcal L,$$
where the expectation is over the profiles $X_m$ and $Y_l$. This shows the local popularity are larger for the convex stochastic order than the global popularities (after re-scaling by the constant factor $L$, which will not impact hit probabilities). Also, the optimal expected hit probability under a given popularity distribution $\mu$ is a convex function of $\mu$. Indeed, in a finite system $\lambda<\infty$, we have 
\[
h^*_\mu(\lambda,\infty)=\frac{1}{\overline\mu\lambda T}\EE_\mu\left[\max_{\substack{\forall m,\:y_m\in \{0,1\}\\\sum_{m\in\mathcal{M}}y_m=C}}\sum_{m\in\mathcal{M}}y_m\mu_m\right],
\]
where the popularity of the items are independently drawn from $\mu$. Convexity follows from maximum being a convex function. 
Then, by definition of the convex stochastic order, we have $\mean{h^*_l(\lambda,\infty)\big|(Z_m)}\geq h^*_\mathcal L(\lambda,\infty)$. Taking the limit as $\lambda\to\infty$ and $\frac{C}{\lambda T}\to\gamma_c$ proves the first statement.

To compute the almost sure limit of $h^*_l(\infty,\infty)$ as $L\to\infty$, the reasoning is very similar as that to compute the limit of $h^*(\infty,T)$ in Theorem~\ref{th:ABT}, except that we do not deal with how to estimate the probabilities. Therefore, we only explain the main steps. Instead of defining a unique threshold $\theta(\gamma_c)$, we now need to define a threshold $\theta^l(\gamma_c,L)$ at each edge cache (we will omit to write the dependency on $\gamma_c$ from now on). Again, for finite values of $L=|\Lcal|$, this threshold is defined by the $\gamma_c$-th upper quantile of the distribution of $(\mu_m^l)_{m\in\Mcal}$ (where the dependence in $L$ is omitted from the notation) for each $l\in\Lcal$. In other words, the local threshold $\theta^l$ is characterized by 
\begin{align*}
\gamma_c&=\prob{\mu_m^l\geq\theta^l(L)}=\mean{\prob{\mu_m^l\geq\theta^l(L)\big|Z_m}}\\
&\underset{L\to\infty}\to\mean{\prob{g(|X_m-Y_l|)\frac{\mu_m^\mathcal L}L\geq\theta^l\Big|Z_m}}\\
&=2\int g^{-1}\left(\frac{L\theta^l}{\mu_m^\mathcal L}\right)dZ_m,
\end{align*}
where the limit follows from the limiting expression of equation~\eqref{eqn: limit local pop}. Intuitively, the local thresholds $(\theta^l(L))_{l\in\Lcal}$ become independent of the particular location $l$ and of the value of $Y_l$ and converge to a same limiting value $\theta^l$ as $L\to\infty$. The inner probability in the expression above can be computed explicitly, due to the particular form of $g$ assumed here, which implies the inverse function $g^{-1}(t)=\min\{x\geq0:\:g(x)\geq t\}$ is well-defined (with $g^{-1}(t)=0$ if $g(x)<t$ for all $x$) and onto:
\[
\prob{\mu_m^l\geq\theta^l\big|Z_m}=g^{-1}\left(\frac{L\theta^l}{\mu_m^\mathcal L}\right),
\]
It remains only to compute the limit of the optimal expected hit probability:
\begin{align*}
h^*_l&\hspace{-0.03in}=\hspace{-0.03in}\frac1{\overline\mu}\mean{\mu_m^l\ind(\mu_m^l\hspace{-0.03in}\geq\theta^l(L))}\hspace{-0.03in}=\hspace{-0.03in}\frac1{\overline\mu}\mean{\mean{\mu_m^l\ind(\mu_m^l\hspace{-0.03in}\geq\theta^l(L))\big|Z_m}}\\
&\underset{L\to\infty}\to\frac1{\overline\mu}\mean{\frac{\mu_m^\mathcal L}L\mean{g(|X_m-Y_l|)\ind(\mu_m^l\geq\theta^l)\Big|Z_m}}
\end{align*}
Again, the inner expectation can be computed explicitly:
\[
\mean{g(|X_m-Y_l|)\ind(\mu_m^l\geq\theta^l)\big|Z_m}=2\int_0^{g^{-1}\left(\frac{L\theta^l}{\mu_m^\mathcal L}\right)} g(t) \:dt,
\]
which yields the claimed expression.
\end{proof}

\end{document}